\def\mode{0}

\if 0\mode

\title{The Collatz process embeds a base conversion algorithm\thanks{Research supported by the European Research Council (ERC) under the EU’s Horizon
2020 research and innovation programme, grant agreement No 772766, Active-DNA project, and Science Foundation
Ireland (SFI), grant number 18/ERCS/5746.}}

\author{Tristan Stérin\thanks{\href{mailto:tristan.sterin@mu.ie}{tristan.sterin@mu.ie}. URL: 
\url{https://dna.hamilton.ie/tsterin/}.}  
\qquad 
Damien Woods\thanks{\href{mailto:damien.woods@mu.ie}{damien.woods@mu.ie}. URL: 
\url{https://dna.hamilton.ie/}.} 
\vspace{2ex} 
\\
Hamilton Institute \\
Department of Computer Science \\ 
Maynooth University  
}
\date{}

\maketitle             
\else

\documentclass[runningheads,envcountsame]{llncs}
\usepackage[british]{babel}
\usepackage[utf8]{inputenc}
\usepackage{graphicx}
\usepackage{xspace}

\spnewtheorem{op}{Open Problem}{\bfseries}{}
\usepackage[hidelinks]{hyperref}

\usepackage{amsmath,amsfonts,amssymb,amsthm}
\usepackage{subcaption}
\newcommand\cqca{CQCA\xspace}
\newcommand\rcqca{rCQCA\xspace}
\newcommand\W{{\mathbb{Z}^2}}
\newcommand\world{\W}
\newcommand\celltype{S}
\newcommand\ntemplate{T}

\newcommand{\Lim}[1]{\raisebox{0.5ex}{\scalebox{0.8}{$\displaystyle \lim_{#1}\;$}}}

\newcommand\bit{s}

\newcommand\car{c}

\newcommand\state{\texttt{\textup{state}}}
\newcommand\seum{\texttt{\textup{total}}}

\newcommand\Pa{\mathcal{P}}

\newcommand\ibin[1]{[\![#1]\!]_2}
\newcommand\itpt[1]{[\![#1]\!]_{3' \to 3}}
\newcommand\ittp[1]{[\![#1]\!]_{3 \to 3'}}
\newcommand\itr[1]{[\![#1]\!]_{3}}
\newcommand\itp[1]{[\![#1]\!]_{3'}}

\newcommand\N{\mathbb{N}}
\newcommand\Z{\mathbb{Z}}

\newcommand\Q{\mathbb{Q}}

\newcommand\Ea{\texttt{\textup{EAST}}}
\newcommand\We{\texttt{\textup{WEST}}}
\newcommand\No{\texttt{\textup{NORTH}}}
\newcommand\So{\texttt{\textup{SOUTH}}}
\newcommand\Ce{\texttt{\textup{CENTER}}}

\newcommand\bdef{half-defined}
\newcommand\fdef{defined}
\newcommand\undef{undefined}

\usepackage{xlop}

\usepackage[textsize=tiny,color=lightgray]{todonotes} 
\usepackage[normalem]{ulem} 

\newcommand{\para}[1]{{\vspace{1.5ex}\noindent\bf #1.}}

\usepackage{comment}

\usepackage{graphicx}
\usepackage{subcaption}
\usepackage{mwe}

\usepackage{thm-restate}

\newcommand{\arxivLemThreeX}{5\xspace}
\newcommand{\arxivLemTwoX}{6\xspace}
\newcommand{\arxivThComplexity}{28\xspace}

\begin{document}

\title{The Collatz process embeds\\ a base conversion algorithm\thanks{Research supported by the European Research Council (ERC) under the EU’s Horizon
2020 research \& innovation programme, grant agreement No 772766, Active-DNA project, and Science Foundation
Ireland (SFI), grant number 18/ERCS/5746. Email: \texttt{tristan.sterin@mu.ie} and  \texttt{damien.woods@mu.ie}}}
\author{Tristan Stérin
\and
Damien Woods
}

\authorrunning{T. Stérin and D. Woods}

\institute{Hamilton Institute \\   
Department of Computer Science\\
 Maynooth University \\
\url{https://dna.hamilton.ie}
}

\maketitle 
\vspace{-4ex}             
\fi

\begin{abstract}

The Collatz process is defined on natural numbers by iterating the map $T(x) = T_0(x) = x/2$ when $x\in\mathbb{N}$ is even and $T(x)=T_1(x) =(3x+1)/2$ when $x$ is odd. In an effort to understand its dynamics, and since Generalised Collatz Maps are known to simulate Turing Machines [Conway, 1972], it seems natural to ask what kinds of algorithmic behaviours it embeds. 
We define a quasi-cellular automaton that  exactly simulates the Collatz process
on the square grid:
on input $x\in\mathbb{N}$, written horizontally in base~2, successive rows give the Collatz sequence of $x$ in base~2. 
We show that vertical columns simultaneously iterate the
map in base 3. 
This leads to our main result: the Collatz process embeds an algorithm that converts any natural number from base 3 to base 2. 
We also find that the evolution of our automaton  computes the parity of the number of 1s in any ternary input. 
It follows that predicting about half of the bits of the iterates $T^i(x)$, for $i = O(\log x)$, is in the complexity class NC$^1$ but outside AC$^0$.
\if 0\mode
Finally, we show that in the extension of the Collatz process to numbers with infinite binary expansions ($2$-adic integers) [Lagarias, 1985], 
our automaton encodes the cyclic Collatz conjecture as a natural reachability problem. \fi
These results show that the Collatz process is capable of some simple, but non-trivial, computation in bases~2 and~3,  suggesting an algorithmic approach to thinking about prediction and existence of cycles in the Collatz process.

\if 1\mode
\
\fi

\end{abstract}


\newcommand{\Nset}{\mathbb{N}}

\section{Introduction}
The Collatz process is defined on natural numbers by iterating the map $T(x) = T_0(x) = x/2$ when $x\in\mathbb{N}$ is even and $T(x)=T_1(x) =(3x+1)/2$ when $x$ is odd.
The Collatz conjecture states that for all $x\in\N\setminus\{0\}$ a finite number of iterations lead to~1.
We know that 1-variable generalised Collatz maps (iterated linear equations of a single natural number variable with arbitrary mod) are capable of running arbitrarily algorithms~\cite{Conway}, modulo an exponential simulation scaling, which motivates our point of view in this paper:  how complicated are the dynamics of the Collatz process? Perhaps the reason this process has resisted understanding is that it embeds algorithm(s) that solve problems with high computational complexity?
Or perhaps by showing that this is not the case, we can get a handle on its dynamics?

We study the structure of iterations of the Collatz process both in binary and in ternary.
We define a 2D  quasi-cellular automaton (\cqca) that consists of a local rule, and a non-local\footnote{The CQCA could easily be altered  to remove the non-local rule and have it be a CA, but the obvious ways to do so would involve using more states and make the mapping to the Collatz process less direct. The \cqca has the freezing property~\cite{goles2015freezing,vollmar1981freezing}; states change obey a partial order, with a constant number of changes permitted per cell.} rule. A natural number is encoded as a binary string input to the \cqca, whose subsequent dynamics exactly execute the Collatz process with one horizontal row per odd iterate. Simultaneously, vertical columns simulate all iterations (both odd and even), but in base 3.

\if 0\mode
  \subsection{Results}
\else
  \para{Results}
\fi
Our main result, Theorem~\ref{th:base_conversion}, is that the \cqca embeds a base conversion algorithm that can convert any natural number $x$ from base~3 to base~2.
This base conversion algorithm is natural and efficient, running in $\Theta(\log x)$ Collatz iterations.
The result puts strict constraints on the short-term dynamics of the Collatz process and  enables us to characterize the complexity of natural prediction problems on the dynamics of the \cqca : predicting about half of the bits of the iterates $T^i(x)$, for $i = O(\log x)$, is in the complexity class NC$^1$ but outside AC$^0$  (see Figure~\ref{fig:influence}).
Our algorithmic perspective is suggestive of an open problem: \emph{small Collatz cycles}: Does there exist $x> 2\in\N$ such that $x = T^{\leq \lceil \log_2 x \rceil}(x)$?
\if 0\mode
  Finally, in the generalisation of the Collatz process to numbers with infinite binary expansion (i.e.\ to $\Z_2$, the ring of $2$-adic integers~\cite{10.2307/2322189}), we show that the \emph{reverse} \cqca can be used to construct the $2$-adic expansion of any rational number involved in a Collatz cycle and that the cyclic Collatz conjecture reformulates as a natural reachability problem.
\fi

The proofs of our main result and supporting lemmas use the fact that the {\em  local} (CA-like) component of the \cqca can be thought of as simultaneously iterating two finite state transducers (FSTs).
One FST computes $x/2$ in ternary on vertical columns, the other FST computes $3x+1$ in binary on  horizontal rows,
as shown in Figure~\ref{fig:fst}. Interestingly, the two FSTs are dual in the sense that states of one are symbols of the other, and arrows of one are read/write instructions of the other. In addition, intuitively, the {\em non-local}  component of the \cqca initiates, or ``bootstraps'', these two FSTs by providing the location of the least significant bit to each.

\if 0\mode
  \subsection{Related and future work}\label{sec:related work}
\else
  \para{Related and future work}
\fi
Since the operations $3x+1$ and $x/2$ have natural base~2 interpretations,
studying the process in binary has been fruitful
\cite{Shallit1992The,DBLP:journals/tcs/Colussi11,DBLP:journals/tcs/Hew16,capco:hal-02062503}.
For example, in binary, predecessors in the Collatz process are characterised by  regular expressions~\cite{1907.00775_RP_2020_sub}: for each $x,k\in\Nset$ there is a regular expression, of size exponential in $k$,   that characterises the binary representations of all $y\in\Nset$ that lead to $x$ via $k$ applications of $T_1$ and any number  of $T_0$.

Cloney, Goles and Vichniac give a unidimensional ``quasi'' CA that simulates the Collatz process~\cite{cloney19873x+}.
Their system works in base 2, for any $x\in\Nset$ given as an input in base 2, successive downward rows encode the iterates $x,T(x),T^2(x),\ldots$, in binary.
The choice of whether to apply $x/2$ and $3x+1$ is done explicitly based on the least significant bit, hence the rule is not local.
In a similar spirit, Bruschi~\cite{nlin/0502061} defines two distinct non-local 1D CA-like rules, one which works in base 2 and the other in base 3.
Finally, in base 6, the Collatz process can be expressed as a local CA because carry propagation does not occur \cite{Korec1992,inpKari_Jarkko13a}.
However, because these systems do not include carries in the automaton state's space, our base conversion result, and parity-based observations on the structure of Collatz iterates, are not apparent, nor is the  \cqca-style embedding of dual FSTs that compute $3x+1$ and $x/2$.

Base conversion is a problem which has been studied in several ways: it is known to be computable by iterating Finite State Transducers \cite{baseChanges}, it is known to be in the circuit complexity class NC$^1$ \cite{hesse2002uniform} and the complexity of computing base conversion of real numbers (infinite expansion) has been explored \cite{Adamczewski2007,baseChanges,Jakobsen2020}. While we know that our framework generalises well to the Collatz process running on infinite binary strings (i.e. the extension~\cite{10.2307/2322189,Monks2004,Rozier2019ParitySO} of $T$ to $\Z_2$, see Remark~\ref{rk:z2}),\if 0\mode\ and we have hints that, in the case of Collatz cycles, our base conversion result applies in that infinite setting (Conjecture~\ref{cj:bc}),\fi\ we leave to future work to show how our base conversion result applies in that case: for instance, can the \cqca convert any element $x\in\Z_3 \cap \Z_2 \cap \Q$ from $\Z_3$ (i.e. infinite ternary expansion) to $\Z_2$ (i.e. infinite binary expansion)?

Generalised Collatz maps, and related iterated dynamical systems, of both one and two variables, simulate Turing machines~\cite{Conway,moore1990unpredictability,moore1991generalized,Koiran1999,DBLP:conf/tamc/KurtzS07,moore2011nature}.
With two variables these maps simulate Turing machines  in real time, just one map application per Turing machine step, and so have a P-complete prediction problem.
The situation is less clear with 1-variable generalised Collatz maps (1D-GCMs), of which Collatz is an instance.
Conway~\cite{Conway} showed that 1D-GCMs simulate Turing machines, but via an exponential slowdown, and it remains as a frustrating open problem whether the simulation can be made to run in polynomial time.\footnote{The problem is closely related to the 2-counter machine problem: when simulating Turing machines, are 2-counter machines exponentially slower than 3-counter machines? See, for example,~\cite{minsky1967computation,Koiran1999,moore2011nature}.}
As with the 2-variable case, Turing Machines simulate 1D-GCMs in polynomial-in-$n$ time (giving an upper bound), but a matching lower bound for predicting $n$ steps of a 1D-GCM on an $n$-digit input remains elusive.
Is the problem in NC, or even in NL? Is it P-complete?
This point of view provides the motivation for a quest to understand the kinds of algorithmic behaviour that might be embedded in 1D-GCMs, and in their prime example, the Collatz map. Here, we show an AC$^0$ lower bound on that prediction problem in the \cqca model, and an NC$^1$ upper bound on roughly half the bits produced over $O(n)$ iterations.
Finding a separation between the Collatz process itself, and  1D-GCMs, is an obvious next step.
The structure we find in Collatz iterations leads us to guess that the \cqca $n$-step prediction problem might be simpler than the analogous problem for 1D-GCMs.
We leave that challenging problem for future work.

We contend that our approach gives a fresh perspective that could lead to progress in understanding the Collatz process. We illustrate with two examples. Firstly, the \cqca  runs in a maximally parallel fashion along a 135$^\circ$ diagonal (see Figure~\ref{fig:colinput}(b)), thus our main base conversion result (Theorem~\ref{th:base_conversion}) implies that such a diagonal encodes a natural number  being converted from base~3 to~2. Hence, it implies that the Collatz process can be interpreted as running along successive \cqca\   {\em diagonals} (rather than rows/columns), giving a new perspective on its dynamics, that we leave to future work. Secondly, by  Theorem~\ref{thm:nc1}, and illustrated in Figure~\ref{fig:influence}, if we pick any cell along a column, the entire upper rectangle (above and to the left of the cell) is tightly characterised by our results: it is a base 3-to-2 conversion diagram (computable in NC$^1$, but not in~AC$^0$). This leaves as future work the lower rectangle (below and to the left) and, we believe, embeds the full complexity of computing $n$ forward Collatz  steps and answering the \textit{small Collatz cycles} conjecture.
Section~\ref{sec:discussion}, Figure~\ref{fig:influence} and Figure~\ref{fig:smallC} contain some  additional discussion.

A simulator, \texttt{simcqca}, was built in order to run the \cqca\footnote{Comprehensive instructions and tutorial at:  \url{https://github.com/tcosmo/simcqca}}. The reader is invited to experience the results of this paper through the simulator.

\begin{figure}[t!] 
  \includegraphics[width=\textwidth]{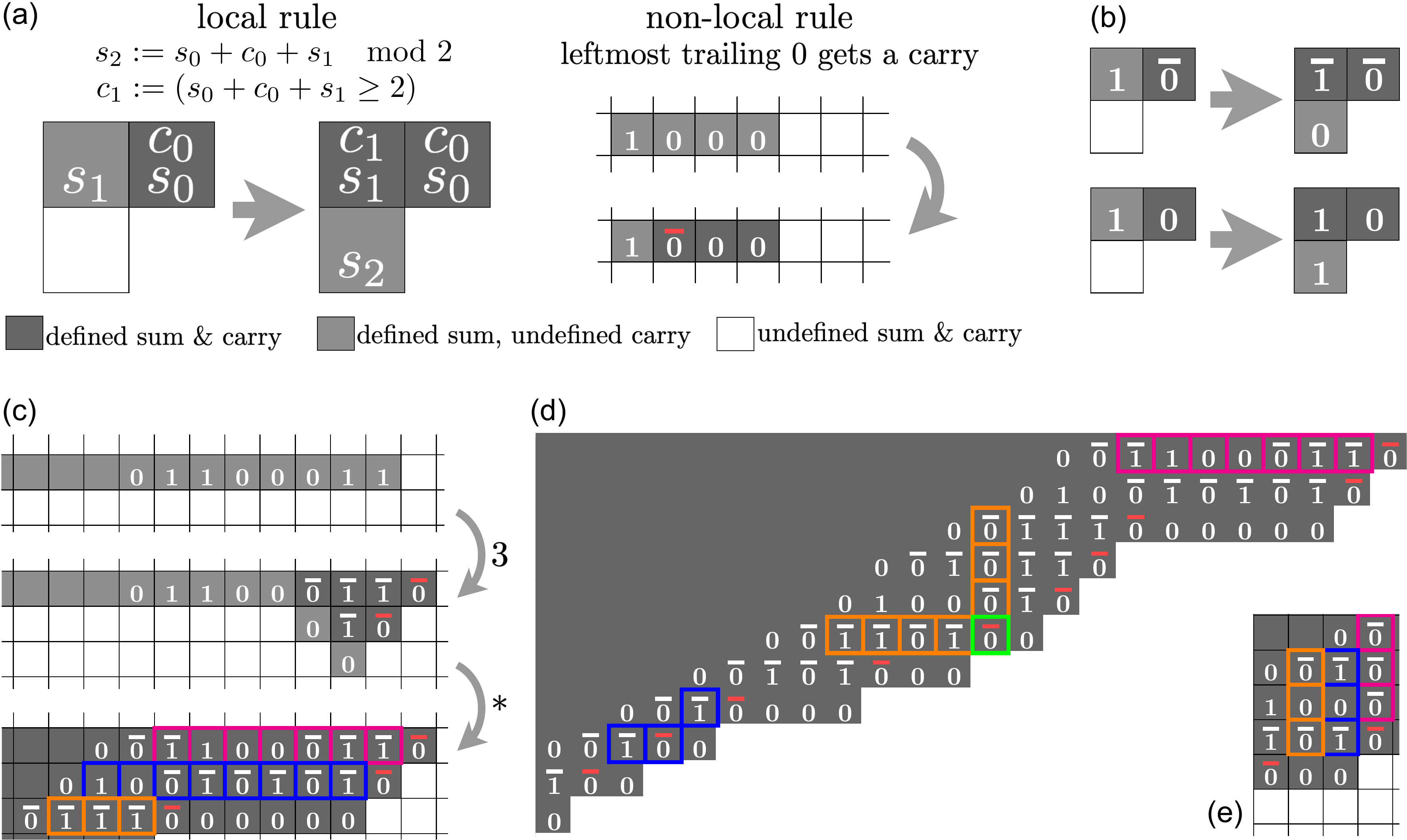}\if 1\mode \vspace{-1.5ex}\fi
  \caption{\cqca definition and examples.  (a)~\cqca local rule, and non-local rule, for sum  bit ($s_2$: written as 0 or 1), and carry bit ($c_1$: over-bar for $c_1=1$, omitted for $c_1=0$), where the bits $c_0$, $s_0$ and $s_1$ are already defined.
  Cell colour distinguishes between carry bit being undefined or $0$, and empty light/dark grey cells have sum bit $0$.
  (b)~Two examples of the local rule.
  (c)~Example run of the \cqca on base~2 (horizontal) input $w = 1100011$,
  showing the initial configuration $c_0[w]$ (Definition~\ref{def:icw}),
  then 3 steps, then $>20$ steps.
  The \cqca evolves in the south-west direction, with successive rows corresponding to odd Collatz iterates in binary (Lemma~\ref{lem:cols}),
  in this case  $(\boldsymbol{99}, \boldsymbol{149}, 224, 112, 56, 28, 14, \boldsymbol{7},\dots)$, i.e.,
  $\ibin{1100011} = 99$ (magenta), $\ibin{10010101} = 149$ (blue) and $\ibin{111} = 7$.
  (d)~Part of the limit configuration $c_\infty[w]$ (Lemma~\ref{lem:lc}) associated to $w = 1100011$. The trivial cycle (1,2), in blue, has been reached. In orange, an instance of the base conversion theorem (Theorem~\ref{th:base_conversion}): north of the green cell (\emph{anchor cell}) reads in base $3'$ (Definition~\ref{def:btp}), $\itp{\bar 0 \bar 0 \bar 0 } = \itr{111} = 13$ and is equal to the base $2$ number represented by the sum bits to the west of the green cell: $\ibin{1101} = 13$.
  (e)~Example run of the \cqca on vertical base 3 input $\alpha = 111$ encoded in base $3'$  as $\bar{0}\bar 0\bar 0$ (Definition~\ref{def:btp}). Each column is a successive $T$-iterate, from magenta to orange, we read: $\itp{\bar 0 \bar 0 \bar 0} = 13$, $\itp{\bar 1  0 \bar 1} = 20 = T(13)$, $\itp{\bar 0 0 \bar 0} = 10 = T(20)$ (Lemma~\ref{lem:cols}), see Figure~\ref{fig:colinput}(a) for a larger vertical example.
  }\label{fig:model}
  \if 1\mode \vspace{-2ex}\fi
\end{figure}

\if 1\mode
\vspace{-2ex}
\section{The Collatz 2D Quasi-Cellular Automaton}\label{sec:2dcqca}
\vspace{-3ex}
\para{The Collatz process} 
\else
\section{The Collatz 2D Quasi-Cellular Automaton}\label{sec:2dcqca}
\subsection{The Collatz process} 
\fi
Let $\N=\{0,1,\dots\}$, $\N^+=\{1,2,\dots\}$ and let $\Z$ be the integers.
The Collatz process consists
of iterating the Collatz map $T:\N \to \N$ where $T(x) = T_0(x) = x/2$ if $x$ is even and $T(x) = T_1(x) = (3x+1)/2$ if $x$ is odd. 
\if 0\mode
The Collatz conjecture states that for any $x\in\N^{+}$ the process will reach~1 after a finite number of iterations. The cyclic Collatz conjecture states that the only strictly positive cycle is $(1,2,1,\dots)$. The Collatz process can be generalised in a natural way to $\Z$ and in that case, $3$ more cycles are known and conjectured to be the only strictly negative cycles: the cycles of $-1$, $-5$ and $-17$. In fact, the Collatz process can be generalised to a much larger class of numbers that includes both $\N$ and $\Z$ but is uncountable: $\Z_2$, the ring of 2-adic integers which syntactically corresponds to the set of infinite binary words \cite{10.2307/2322189}. Given that generalisation, the Collatz process can be run on more exotic numbers such as $-\frac{4}{23} \in \Z_2 \cap \Q$ and the Collatz conjecture generalises as follows: all rational 2-adic integers eventually reach a cycle\footnote{The set $\Z_2\cap \Q$ exactly corresponds to irreducible fractions with odd denominator and parity is given by the parity of the numerator. For instance, the first Collatz steps of $-\frac{4}{23}$ are: $(-\frac{4}{23},-\frac{2}{23},-\frac{1}{23},\frac{10}{23},\frac{5}{23},\,\dots)$. It reaches the cycle $(\frac{5}{23},\frac{19}{23},\frac{40}{23},\frac{20}{23},\frac{10}{23},\frac{5}{23}$ \dots).\\} (known as Lagarias' Periodicity Conjecture \cite{10.2307/2322189}). While, in Sections~\ref{sec:2dcqca} and~\ref{sec:base conversion} of this paper, we are mainly concerned with finite binary inputs (i.e.\ representing elements of $\N$) our framework can naturally generalize to infinite binary inputs (i.e.\ representing elements of $\Z_2$), see Remark~\ref{rk:z2}.
\else
The Collatz conjecture states that for any $x\in\N$ the process will reach~1 after a finite number of iterations. The cyclic Collatz conjecture states that the only strictly positive cycle is $(1,2,1,\dots)$. 
\fi

\if 0\mode
\subsection{Informal \cqca definition}
\else
\para{\cqca definition} 
\fi
The \cqca  is pictorially defined in Figure~\ref{fig:model} and more formally defined in Section~2.3\if 0\mode. \else\ of the full version of this paper \cite{Collatz2arxiv}.\fi\  
A {\em configuration} is defined on $\Z^2$, where each cell in $\Z^2$ has a state $(s,c) \in S = \{ 0,1,\bot\}^2 \setminus \{(\bot,0),(\bot,1)\} $
containing \emph{sum bit} $s$ and \emph{carry bit} $c$, 
said to be {\em defined} if 0/1 or {\em undefined} otherwise ($\bot$).
We say that the cell is {\em defined} 
if both the sum and carry defined, 
{\em half-defined} if sum is defined and carry undefined, and
{\em undefined} if sum and carry are undefined. 
Note, in all Figures, cell colour distinguishes between a carry bit being undefined or being $0$, and empty light/dark grey cells have sum bit $0$, as defined in Figure~\ref{fig:model}(a).

A configuration update step is parallel and synchronous:  
First, the non-local rule is applied (at most 1 row per step is updated by the non-local rule) then, on the updated configuration, the local rule is applied everywhere it can be.
The non-local rule  implements the $+\,1$ part of the $3x+1$ operation as follows: 
on any horizontal row which has only half-defined cells, with exactly one cell $\rho$ having sum bit 1, then the neighbour immediately to the right of $\rho$ 
gets, \textit{ex nihilo}, a carry bit equal to $1$ (shown in red in Figure~\ref{fig:model}(a), see~\if 0\mode Section~2.3\else\cite{Collatz2arxiv}\fi \ for a formal definition). 
The local rule works as shown in Figure~\ref{fig:model}(a): for any undefined cell $e\in\W$, with a half-defined north neighbour with sum $s_1$, and defined north-east neighbour with sum $s_0$ and carry $c_0$, $e$'s sum bit becomes $s_2 := s_0+c_0+s_1 \mod 2$.  Simultaneously, the carry, $c_1 := (s_0+c_0+s_1\geq 2)$, is placed on the cell to the north of $e$. 
Figure~\ref{fig:colinput}(b) shows that the ``natural'' evolution frontier of the \cqca is along a $135^\circ$ diagonal. 
\if 1\mode Let the unit cardinal vectors in $\Z^2$ be $\Ea=(1,0),\,\We=(-1,0),\,\No=(0,1),\,\So=(0,-1)$.\fi

\if 0\mode
\subsection{Formal \cqca definition}\label{sec:formal_def}
\if 0\mode Let the unit cardinal vectors in $\Z^2$ be $\Ea=(1,0),\,\We=(-1,0),\,\No=(0,1),\,\So=(0,-1)$ and let $\Ce =(0,0)$.\fi

\begin{definition}[\cqca]
\normalfont 
The \cqca is defined by the 7-tuple $\langle \celltype, \world, \ntemplate_\text{local}, f_\text{local}, \ntemplate_\text{non-local}, f_\text{non-local}, I \rangle$ where:
\begin{itemize}
\item $\celltype = \{ 0,1,\bot\}^2 \setminus \{(\bot,0),(\bot,1)\}$ is the state space, the first component of each pair represents a sum bit and the second a carry bit;
\item $\world$ is the evolution lattice;
\item $\ntemplate_\text{local} = (\Ce, \Ea, \So, \No, \No+\Ea)$ is the local neighbourhood template;\footnote{$\ntemplate_\text{local}$ is a sub-neighbourhood of the Moore neighbourhood.}
\item $f_\text{local}: \celltype^{5} \to \celltype$ is  the local evolution rule;
\item $\ntemplate_\text{non-local} = (\ldots, 2\times\We, \We, \Ce, \Ea, 2 \times \Ea, \ldots)$ is  the non-local neighbourhood template; 
\item $f_\text{non-local}: \celltype^{\mathbb{Z}} \to \celltype$ is the non-local evolution rule;
\item $I = \{ c_0[w] \mid w \in \{0,1\}^* \} \cup \{ c_0[\alpha] \mid \alpha \in \{0,1,2\}^* \}$ is the set of valid initial configurations which are described in Definition~\ref{def:icw} and Definition~\ref{def:ica}.

\end{itemize}

We say that the state $(s,c)\in\celltype$ is: \emph{\undef}\ if $\bit = \bot$ and $\car = \bot$, \emph{\bdef}\ if $\bit \neq \bot$ and $\car = \bot$,  and \emph{\fdef}\ if $\bit \neq \bot$ and $\car \neq \bot$.

The local rule $f_\text{local}(s_\text{center}, s_\text{e}, s_\text{s}, s_\text{n}, s_\text{ne})$  is defined according to three cases:

\begin{enumerate}

    \item If $s_\text{center}=(\bit,\bot)$ is \bdef:
    \begin{enumerate}
        \item if $s_\text{e}=(\bit',\car')$ is \fdef\ then the local rule returns $(\bit,\car'')$ with $\car'' = 1$ if $\bit+\bit'+\car' \geq 2$ else $\car'' = 0$ \textbf{(Carry propagation)}; 
        \item else the local rule returns $s_\text{center}$.
    \end{enumerate} 
    \item Else if $s_\text{center}=(\bot,\bot)$ is \undef: 
        \begin{enumerate}
        \item if $s_\text{n} = (\bit,\bot)$ is \bdef\, and $s_{\text{ne}}=(\bit',\car')$ is \fdef\, then the local rule returns $((\bit+\bit'+\car')\%2,\bot)$ with $\%$ the modulo operator \textbf{(Forward deduction)};
        \item else the local rule returns $s_\text{center}$.
    \end{enumerate}
    \item If $s_\text{center} = (\bit,\car)$ is \fdef: the local rule returns $s_\text{center}$.

\end{enumerate}

The non local rule $f_\text{non-local}(\dots,\, s_\text{2w}, s_\text{w},s_\text{center}, s_\text{e},s_\text{2e},\,\dots)$ is defined according to three cases:

\begin{enumerate}

    \item If $s_\text{w} = (1,\bot)$ and $s_\text{center}=(\bit',\bot)$ with $\bit' \in \{0,\bot\}$ and, for all $i\geq 1$, $s_{i \times \text{e}} = (\bit'', \bot)$ with $\bit'' \in \{0,\bot\}$ then the rule returns $(0,1)$ \textbf{(Bootstrapping, or `+ 1')};
    \item else if $s_\text{w} = (0,\bot)$ and $s_\text{center}=(0,\bot)$ and, there exists $i\geq 1$ such that $s_{i \times \text{w}} = (1,\bot)$ then the rule returns $(0,0)$;
    \item otherwise, the rule returns $s_\text{center}$.

\end{enumerate}

A configuration of the \cqca\ is a function $c: \world \to \celltype$. The evolution of c is given by successive iterations, with each iteration first applying the
non-local rule, and then applying the local rule as follows. In this paper, we consider only evolution of initial configurations given by inputting a row/column to the \cqca: configurations in the set $I = \{ c_0[w] \mid w \in \{0,1\}^* \} \cup \{ c_0[\alpha] \mid \alpha \in \{0,1,2\}^* \}$ (Definition~\ref{def:icw} and Definition~\ref{def:ica}). The evolution of $c_0\in I$ is inductively defined as follows, for $i\in \N$, if $c_i$ is the $i^\text{th}$ iterate of the \cqca\ starting from $c_0$, then $c_{i+1}$ is defined for all $w \in \world$ by: 
\begin{align*}
c'_{i}(w)  &= f_\text{non-local}(\dots,\,c_{i}( w+\eta_{-1}),c_{i}(w+\eta_0),c_{i}(w+\eta_1),\dots)\\
c_{i+1}(w) &= f_\text{local}(c'_{i}(w+\eta'_0),\dots,c'_{i}(w+\eta'_{4}))
\end{align*}
Where $(\dots,\,\eta_{-1},\eta_0,\eta_1,\dots) = \ntemplate_\text{non-local}$ and $(\eta'_0,\dots,\eta'_{4}) = \ntemplate_\text{local}$.

We write $F: S^\W \to S^\W$ to denote the transition function of the \cqca\ on configurations.

\end{definition}

\begin{remark}
A simulator for the \cqca, \texttt{simcqca}, was built to support this research. It can be found here: \url{https://github.com/tcosmo/simcqca} (release \texttt{v0.4} is stable). A comprehensive tutorial is present in \texttt{README.md} which allows one to follow all the results of this paper (and more) interactively. The sources (C++) are documented and can be freely modified (MIT License). Note that the colour convention regarding cells' background in the simulator is the opposite of the paper: darker for undefined, in between for half defined and lighter for defined. Also, for efficiency, some cells which in theory are defined and have value $(0,0)$, are simply left  undefined (i.e. value $(\bot,\bot)$) in the simulator.
\end{remark}
\fi

\begin{figure}[t]
        \centering
        
        \begin{subfigure}[t]{0.475\textwidth}  
            \centering 
            \includegraphics[scale=0.5]{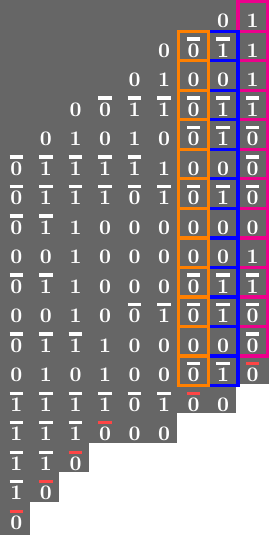}
           \caption[]%
            {{\small Columns in the \cqca iterate the Collatz process in base $3'$. }}
        \end{subfigure}
        \hfill 
        \begin{subfigure}[t]{0.475\textwidth}
            \centering
            \includegraphics[scale=0.6]{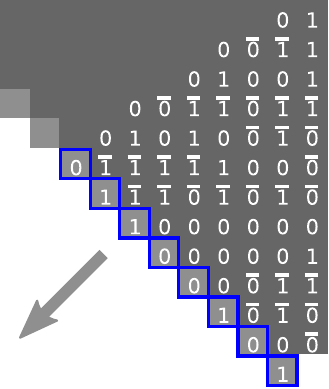}
            \caption[]%
            {{\small The evolution frontier of the \cqca is a 135$^\circ$ diagonal.}}    
        \end{subfigure}
        
\if 1\mode \vspace{-1ex} \fi
        \caption[]
        {Evolution of a column input in the \cqca.  Colours as in Figure~\ref{fig:model}. 
        (a) Portion of $c_\infty[\alpha]$, the limit configuration of the \cqca on column input $\alpha=111211101211$. Successive columns iterate the Collatz process in base $3'$  (Lemma~\ref{lem:cols}). We read: $\itp{\bar 1 0 \bar 1 \bar 1 0 \bar 1 0 0 \bar 1 \bar 1 0 \bar 1} = 408314 = T(272209)$ (blue) and $\itp{\bar 0 0 \bar 0 \bar 0 0 \bar 0 0 0 \bar 0 \bar 0 0 \bar 0} = 204157 = T(408314)$ (orange). 
        (b) Evolution of $c_0[\alpha]$ after $|\alpha|$ \cqca steps, highlighting the ``natural'' frontier of evolution of the \cqca as a 135$^\circ$ diagonal (blue cells). Cells to the north-east of the diagonal are defined, cells to the south-west are undefined and cells on the diagonal are half-defined.  
}\label{fig:colinput}
    \end{figure}

\if 0\mode
\subsection{Transducer simulation, base 2 and base 3$'$}
\fi

\if 0\mode
The \cqca rule has a local and non-local component.   
The {\em local} component simulates two FSTs:  
the $3x+1$ FST in binary along horizontal rows and the $x/2$ FST in ternary along vertical columns (Figure~\ref{fig:fst}). 
Intuitively, the {\em non-local} component of the \cqca provides the least significant bit to these FSTs, in other words, it runs $x/2$ in binary (removing a trailing 0) and $3x+1$ in ternary (adding a trailing $1$). 
Interestingly, these two FSTs are closely related, we say they are \textit{dual}: states of one are symbols of the other and that arrows of one are read/write instructions of the other. The proof of our main result Theorem~\ref{th:base_conversion}, and of Lemmas~\ref{lem:rows} and~\ref{lem:cols}, use the ability of the \cqca to simulate these 
FSTs simultaneously, one horizontally and the other vertically. 
\fi

\para{Base $2$, $3$ and $3'$} 
Let $\{0,1\}^*$ be the set of finite binary strings, $\{0,1,2\}^*$ be the set of finite ternary strings.
We  index strings from their rightmost symbol, and $|\cdot|$ denotes  string length, meaning we write, for instance, $w = w_{|w|-1}\dots w_1 w_0 \in \{0,1\}^*$. We use the standard interpretation of strings from $\{0,1\}^*$ and $\{0,1,2\}^*$ as, respectively, base 2 and base 3 encodings of natural numbers where the rightmost symbol is the least significant digit. 
We write $\ibin{\cdot}: \{0,1\}^* \to \N$ and $\itr{\cdot}: \{0,1,2\}^* \to \N$ for those interpretations. For instance, $\ibin{110} = \itr{20} = 6$. 

The \cqca uses a particular encoding for base $3$ strings, called base $3'$, over the four-symbol alphabet $\{0,\bar 0, 1, \bar 1\}$. 
The \cqca states $(0,0)$, $(0,1)$, $(1,0)$,  $(1,1)$ $\in S$ respectively represent $0,\bar 0,1,\bar 1$, using a (sum-bit, carry-bit) notation.
The function $\itpt{\cdot} :  \{0,\bar 0, 1, \bar 1\}^\ast\rightarrow \{0,1,2 \}^\ast$  converts  base $3'$ to base $3$ in a straightforward symbol-by-symbol fashion: $0\mapsto 0$, $\bar 0 \mapsto 1$, $1 \mapsto 1$ and $\bar 1 \mapsto 2$. 
For instance, $\itpt{\bar{0}\bar{0}01\bar{1}} = 11012$. We write $\itp{\cdot} = \itr{\itpt{\cdot}}$ as the interpretation of base $3'$ strings as natural numbers. However, because there are two choices for encoding the trit $1$ in base $3'$, converting from base $3$ to base $3'$ requires a definition:

\begin{definition}[Base 3 to \ensuremath{3'} encoding]\label{def:btp}
\normalfont
The function $\ittp{\cdot} : \{0,1,2\}^\ast\rightarrow \{0,\bar 0, 1, \bar 1\}^\ast$ 
encodes a base 3 word $\alpha$ as a base $3'$ word as follows: 
$0$ is encoded as $0$,  
$2$ is encoded as $\bar{1}$, and 
$1$ is encoded as $1$ when the rightmost neighbour different from $1$ is a $2$, and as $\bar{0}$ otherwise. 
E.g., $\ittp{1112 1110 12 11} = 111\bar{1} \bar{0}\bar{0}\bar{0}0 1\bar{1} \bar{0}\bar{0}$.
\end{definition}

\if 0\mode
\begin{remark}\label{rk:btp}
Here we justify the $\ittp{\cdot}$ encoding of Definition~\ref{def:btp}. In limit configurations of the \cqca (Lemma~\ref{lem:lc}), if the state at position $e\in\W$ is $(0,1)$ then, the state at position $e+\Ea$ is $(1,1)$ since any other choice would not propagate a carry to position $e$. Hence, the sum bit of position $e+\So$ is $0+1+1\mod 2 = 0$. As a consequence, in base 3$'$ readings of vertical \cqca columns, the symbol $\bar 0$ can only be followed by another $\bar 0$ or a $0$ and with a similar argument we get that the symbol $1$ can only be followed by another $1$ or a $\bar 1$.
\end{remark}
\fi

\if 0\mode
The proofs of the following two lemmas, which assert the correctness of the $3x+1$ and the $x/2$ FSTs are rather straightforward and left, together with more details about these FSTs, to Appendix~\ref{app:fst}.

\begin{restatable}[The $3x+1$ FST]{lemma}{lemT}\label{lem:fst3x}
\normalfont
Let $x\in\N$ and $w = w_{|w|-1}\dots w_0\in \{0,1\}^*$ the standard binary representation of $x$ (least significant bit is $w_0$). Then, given initial state $\bar 0$ and input $00w$ where bits are read starting at $w_0$ onwards, the $3x+1$ FST outputs, from least significant to most significant bit, $w' \in \{0,1\}^*$ (of length $|w|+2$) such that $\ibin{w'} = 3x+1$.
\end{restatable}

\begin{restatable}[The $x/2$ FST]{lemma}{lemD}\label{lem:fstx2}
\normalfont
Let $x\in\N$ and $\gamma = \gamma_{|\gamma|-1} \dots \gamma_{0} \in \{0,\bar 0,1,\bar 1\}^*$ be the standard base $3'$ represention of $x$ (Definition~\ref{def:btp}), with potential leading $0$s (least significant symbol is $\gamma_0$). Then, given initial state $0$ and input $\gamma$ where symbols are read starting at $\gamma_{|\gamma|-1}$ \emph{downwards}, the $x/2$ FST outputs $\gamma'$, from most significant to least significant symbol, such that: $\itp{\gamma'} = \lfloor \frac{x}{2} \rfloor$. Finally, if the $x/2$ FST ends in state $0$ after reading the symbols of $\gamma$ then $x$ is even and odd otherwise.
\end{restatable}
\fi

\begin{figure}[t]
\centering
\includegraphics[scale=0.8]{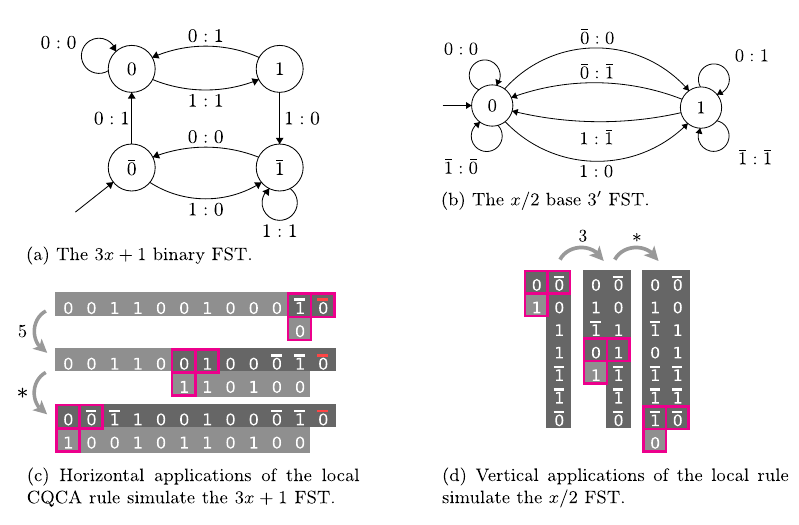} \if 1\mode \vspace{-2ex}\fi
        \caption[]{Panels (a) and (b) present the $3x+1$ binary Finite State Transducer (FST) and the $x/2$ base $3'$ FST. Instructions ``$s_1:s_2$'' mean ``read $s_1$, write $s_2$''. The FSTs are \emph{dual} to one another: states of one are symbols of the other and arrows of one are read/write instructions of the other. 
Panels (c) and (d) show the relation with the local \cqca rule:  (c)~horizontal \cqca applications correspond to simulating the $3x+1$ FST and (d) vertical \cqca applications to simulating the $x/2$ FST. The same colour code as Figure~\ref{fig:model} is used. More precisely, in (c),  on input $\ibin{00110010001} = 401$ and initial state $\bar{0}$ (rightmost $0$ with red carry), we get output $\ibin{10010110100} = 1204 = 3\times 401+1$. In (d) on input $\itp{\bar 0 0 1 1 \bar 1 \bar 1 \bar 0} = 862$ and initial state $0$ (top-left most sum bit $0$), we get output $\itp{01\bar{1}0\bar{1}\bar{1}\bar{1}} = 431 = 862/2$. 
For clarity of exposition, we are ``illegally'' running the \cqca in a vertical (c), or horizontal (d),   sequential mode --- in the legal,  or ``natural'', parallel mode, see Figure~\ref{fig:colinput}(b), more cells would have been computed than are shown. The non-local component of the \cqca rule provides the FSTs with the (red) carry at the least significant digit.}
        \label{fig:fst}
    \end{figure}

\if 1\mode
\para{Transducers and Duality}
The \cqca rule has a local and non-local component.   
The {\em local} component simulates two FSTs:  
the $3x+1$ FST in binary along horizontal rows and the $x/2$ FST in ternary along vertical columns (Figure~\ref{fig:fst}). 
Intuitively, the {\em non-local} component of the \cqca provides the least significant bit to these FSTs, in other words, it runs $x/2$ in binary (removing a trailing 0) and $3x+1$ in ternary (adding a trailing $1$). 
Interestingly, these two FSTs are closely related, we say they are \textit{dual}: states of one are symbols of the other and that arrows of one are read/write instructions of the other. The proof of our main result Theorem~\ref{th:base_conversion}, and of Lemmas~\ref{lem:rows} and~\ref{lem:cols}, use the ability of the \cqca to simulate these 
FSTs simultaneously, one horizontally and the other vertically. 
\fi

\section{Base conversion in the Collatz process}\label{sec:base conversion}

\begin{figure}[t!] 
\centering
\includegraphics[width=1\textwidth]{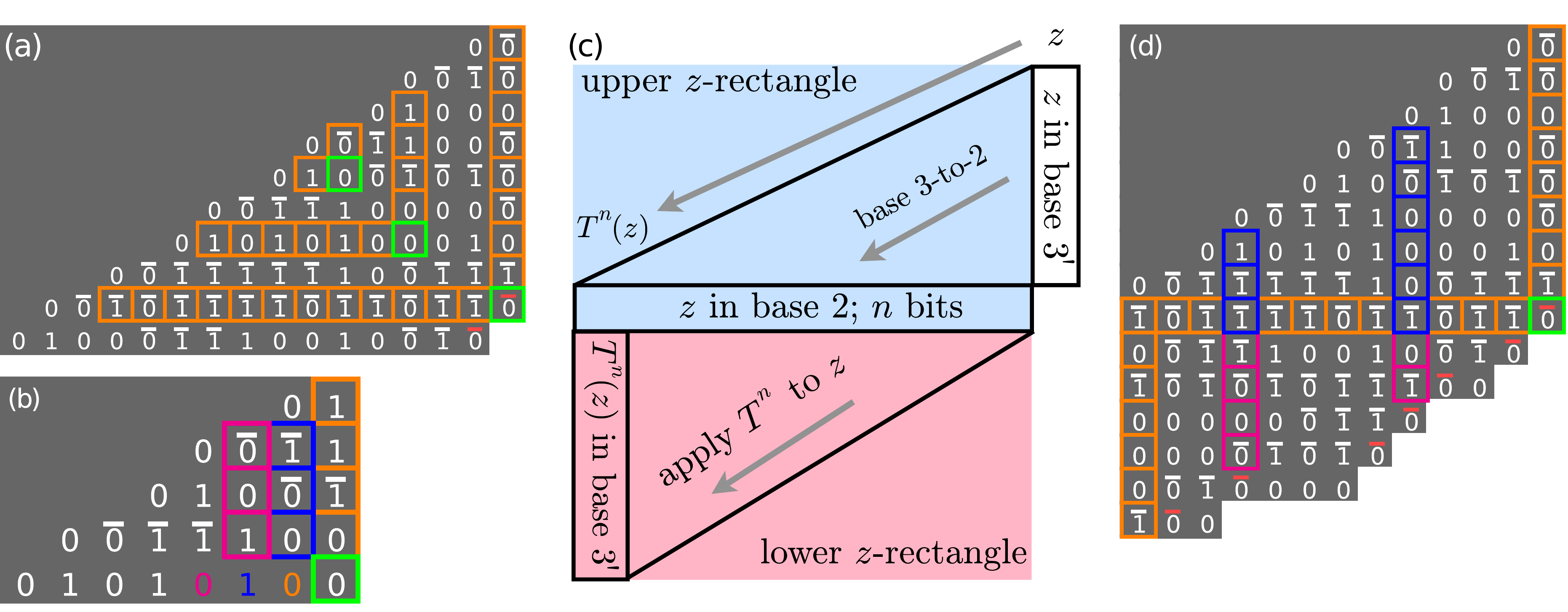}\if 1\mode \vspace{-3ex}\fi
\caption{Structure of the \cqca, and Collatz iterations, implied by our results. 
(a)~Three instances of Theorem~\ref{th:base_conversion}. From the innermost to the outermost they read $\itp{\bar{0}}= \itr{1} = \ibin{1} = 1$, then $\itp{11\bar{1}0} = \itr{1120} = \ibin{101010} = 42$, and finally $\itp{\bar 0 \bar 0 0 \bar 0 \bar 0 \bar 0 0 \bar 1} = \ibin{101111011011} = 3035$. Anchor cells are in green.
(b)~ Zoom-in of base conversion diagram of $\itp{11\bar{1}0}$ (middle example in (a)). Each column,  excluding its bottom cell, represents successive integer divisions by $2$, e.g., 
from orange to magenta columns we read  $\itr{1120} = 42$, $\itr{210} = 21$ and $\itr{101} = 10$. At each step, the sum bits of the bottom row ``store'' the parity information of the previous column that was divided by $2$: the orange bit  gives the parity of the orange column and so on.  
Checking parity in base $3'$ is checking the parity of the number of $1$s (both sum and carry bits), hence outside AC$^0$ (Corollary~\ref{lem:parityCheck},  
 Theorem~\ref{thm:not in AC0}).
(c)~For any input $z \in \N$ the schematic shows   $n= \lceil \log_2 x \rceil$ iterations of the Collatz map $T$. 
The values  $x, T(x), T^2(x), \ldots, T^n(x)$  appear as $n$ columns, written in ternary (base~$3'$). 
The configuration is cut by a horizontal line, whose cells encode $z$ in binary (Theorem~\ref{th:base_conversion}). 
The `base conversion' upper $z$-rectangle is simple to define in terms of $z$, and is computable in NC$^1$. 
The lower $z$-rectangle embeds the full complexity of computing $n$ iterations of $T$, but with  input being in base~2 and output in base~3$'$. 
(d)~The influence of $z = \itp{\bar 0 \bar 0 0 \bar 0 \bar 0 \bar 0 0 \bar 1} = \itr{11011102} = 3035$ on its next $n = \lceil \log_2 z \rceil$ Collatz iterates. The column in orange on the right is the base $3'$ encoding of~$z$. By Lemma~\ref{lem:cols}, the column in orange on the left is the base $3'$ encoding of $T^n(z) = T^{12}(z) = \itr{2020002} = 1622$. By Theorem~\ref{th:base_conversion}, the orange horizontal sum bits give the base 2 representation of $z$. Hence the cells outlined in blue (upper $z$-rectangle) in $T^4(z)$ (second outlined column to the right) and $T^9(z)$ (third outlined column), are entirely determined by the base $3'$ to base 2 conversion diagram of $z$, only the cells outlined in magenta (lower $z$-rectangle) are independent from the base conversion algorithm.}
\label{fig:influence}
\label{fig:cqca_idea}
\if 1\mode \vspace{-3ex} \fi
\end{figure}

\begin{definition}[Binary initial configuration \ensuremath{c_0[w]}]\label{def:icw}
\normalfont
For any binary input $w\in\{0,1\}^*$, we define $c_0[w]\in S^\W$ to be the initial configuration of the \cqca\ with $w$ written on the horizontal ray $y=0, x < 0 $ as follows: 
for $1 \leq i \leq |w|$ we set $c_0[w](-i,0) = (w_{i-1},\bot)$, 
for $i > |w|$ we set $c_0[w](-i,0) = (0,\bot)$
and for all other positions $(x,y)\in\W$ we set $c_0[w](x,y) = (\bot,\bot)$. 
\end{definition}

\begin{definition}[Ternary  initial configuration \ensuremath{c_0[\alpha]}]\label{def:ica}
\normalfont
For any ternary input $\alpha\in\{0,1,2\}^*$ we define $c_0[\alpha]\in S^\W$ to be the initial configuration of the \cqca 
with $\alpha' = \ittp{\alpha}$ (Definition~\ref{def:btp}) written on the vertical ray $x=0, y > 0$ as follows: 
for $1 \leq i \leq |\alpha|$ 
we set $c_0[\alpha](0,i) = \state(\alpha'_{i-1})$ where~$\state: \{0,\bar 0, 1, \bar 1\} \to S$  
is such that $\state(0) = (0,0)$, $\state(\bar 0) = (0,1)$, $\state(1) = (1,0)$ and $\state(\bar 1) = (1,1)$. 
Also, for all $x < 0$ 
we set $c_0[\alpha](x,|\alpha|) = (0,\bot)$ 
and finally at all other positions we set $c_0[\alpha](x,y) = (\bot,\bot)$.
\end{definition}

Each initial configuration has a well-defined unique limit configuration: 
\begin{lemma}[Limit configurations \ensuremath{c_\infty[w]} and \ensuremath{c_\infty[\alpha]}]\label{lem:lc}
\normalfont
Let $w\in\{0,1\}^*$ be a finite binary string and $\alpha\in\{0,1,2\}^*$ be a finite ternary string. 
Then, in the \cqca evolution from the initial configurations $c_0[w] \in S^\W$, and $c_0[\alpha]  \in S^\W$, both sum and carry bits in any cells are final: if set, they never change.
Hence, limit configurations $c_\infty[w] = \Lim{i \to \infty} F^i(c_0[w])$ and $c_\infty[\alpha] = \Lim{i \to \infty} F^i(c_0[\alpha])$ exist.\footnote{Furthermore, although the following fact is not used in the rest of this paper, one can  prove that limit configurations contain no half-defined cells: each cell is either defined or undefined.}
\end{lemma}
\begin{proof}
Both the local and non-local rules of the \cqca can only be applied either on undefined cells or half-defined cells. Furthermore, if applied on an undefined cell, the cell becomes half-defined or defined, and when applied on a half-defined cell it becomes defined. Hence, the following partial order $(\bot,\bot) < (s,\bot) < (s,c)$ with $s,c \in \{0,1\}$ holds on the states of the \cqca: it has the freezing property~\cite{goles2015freezing,vollmar1981freezing} and cells can change state at most twice. Limit configurations are well-defined by taking the final state of each cell.
\end{proof}

\begin{example}
Figure~\ref{fig:model}(c) (top) and (d) respectively show $c_0[w]$ and a portion of $c_\infty[w]$ for $w=1100011$. Figure~\ref{fig:colinput}(a)  shows a portion of $c_\infty[\alpha]$ for $\alpha=111211101211$;  by Definition~\ref{def:btp} we have $\ittp{1112 1110 12 11} = 111\bar{1} \bar{0}\bar{0}\bar{0}0 1\bar{1} \bar{0}\bar{0}$.
\end{example}

Next, we define how to read base 2 strings on rows and base $3'$ on columns of the \cqca: 

\begin{definition}[Mapping rows and columns to natural numbers]\label{def:read}
\normalfont
Let $x_0,y_0 \leq 0$ and $c\in S^\W$ be a configuration. 
A finite segment of defined cells along a horizontal row $x_0$ of $c$ is said to {\em give} word $w \in\{0,1\}^*$ if $w$ is exactly the concatenation of the sum bits in these cells (LSB on right). 
An infinite horizontal row $y_0$ of $c$ is said to \emph{give}  $w\in\{0,1\}^\ast$ 
 if there is a $k\geq 0$ such that the defined sum bits on $y_0$ are $0^\infty w  0^k $.
 
A contiguous segment of defined cells on the column of $c$ with $x$-coordinate $x_0$  is said to \emph{give} the base $3'$ word $q\in\{0,\bar 0,1, \bar 1\}^\ast$ if $q$ is exactly the concatenation of the base $3'$ symbols corresponding to each cell's state\footnote{The \cqca defined states are $(0,0),(0,1),(1,0),(1,1)$ and they respectively map to base $3'$ symbols $0,\bar 0,1, \bar 1$.}   (where the southmost cell gives the least significant trit).
\end{definition}

Intuitively, the following lemma says that rows of \cqca\ encode odd Collatz iterates in binary. There is one odd iterate per row, Figure~\ref{fig:model}(c).

\begin{lemma}[Rows simulate Collatz in base 2]\label{lem:rows}
\normalfont 
Let $z\in\N^+$ and let 
$w\in\{0,1\}^* \setminus \{ 0\}^*$ 
be the standard base $2$ representation of $z$ 
and, by Lemma~\ref{lem:lc}, let $c_\infty[w]\in S^\W$ be the \cqca limit configuration on input $w$. 
Then, the horizontal row $y_0\leq 0$ of $c_\infty[w]$ gives the base 2 representation of the ${|y_0|}^\text{th}$ odd term in the Collatz sequence of $z$ (Definition~\ref{def:read}).
\end{lemma}
\begin{proof}
Note that it is enough to show that row $y=-1$ of $c_\infty[w]$ gives the base $2$ representation of the second odd term in the Collatz sequence of $z$ and then inductively apply the argument to all rows $y < -1$ to get the result. 
We have $w\neq 0^n$, hence there is at least one $1$ in $w$ and so, on row $y=0$ of $c_\infty[w]$, the non-local rule of the \cqca (Figure~\ref{fig:model}(a)) was applied exactly once, at position $x_0 = \max(\{x < 0 \mid c_\infty[w](x,0) = (s,c) \text{ with } s = 1\}) + 1$ and we have $c_\infty[w](x_0,0) = (0,1)$. Sum bits on row $y=0$ up to column $x_0-1$ give the representation of $z' \in \N$ (Definition~\ref{def:read}), the first odd term in the Collatz sequence of $z = \ibin{w}$.

From position $x=x_0-1$ down to $x_0 = -\infty$, the local rule of the \cqca was applied to produce carries on row $y=0$ and sum bits on row $y=-1$. Given the definition of the \cqca local rule, that computation corresponds exactly to applying the binary $3x+1$ FST (Figure~\ref{fig:fst}(a) and (c)): input is read in the sum bits of row $y=0$ from $x = x_0-1$ down to $x = -\infty$, output is produced in the sum bits of row $y=-1$, the initial state is $\bar 0$, corresponding to $c_\infty[w](x_0,0) = (0,1)$.  Hence, \if 0\mode by Lemma~\ref{lem:fst3x},\else by Lemma~\arxivLemThreeX in the full version of this paper \cite{Collatz2arxiv},\fi \ which asserts the correctness of the $3x+1$ FST computation, sum bits of row $y=-1$ from $x=-\infty$ to $x=x_0-1$ give the binary representation of $3z'+1$; starting with a $(0)^\infty$ prefix and with LSB to the east. By ignoring the $m \geq 1$ trailing zeros on row $y=-1$ (there is at least one because $3z'+1$ is even), we get the representation of $(3z'+1)/2^m$, the second odd iterate in the Collatz sequence of $z$. Note the non $(0)^\infty$ part of row $y=-1$ is produced in  $\leq |w|+2$ 
steps in the \cqca.
\end{proof}

\begin{remark}\label{rk:row1}
Although Lemma~\ref{lem:rows} only deals with odd Collatz iterations, even Collatz iterations also naturally appear in the \cqca: even terms occurring between the $i^\text{th}$ and the $(i+1)^\text{th}$ odd Collatz iterates correspond to the trailing $0$s on row $y=-(i+1)$ of the \cqca.
For example, on the third row in Figure~\ref{fig:model}(c) (bottom) we read $7 = \ibin{111}$ but also, in the trailing $0s$, all $2^n \cdot 7$ for $1 \leq n \leq 6$.
\end{remark}

\begin{remark}\label{rk:z2}
\if 1\mode
The Collatz process can be generalised to an uncountable class of numbers that includes both $\N$ and $\Z$: $\Z_2$, the ring of 2-adic integers which syntactically corresponds to the set of infinite binary words \cite{10.2307/2322189,caruso:hal-01444183}. Given that generalisation, the Collatz process can be run on more exotic numbers such as $-\frac{4}{23} \in \Z_2 \cap \Q$ and the Collatz conjecture generalises as follows: all rational 2-adic integers eventually reach a cycle\footnote{The set $\Z_2\cap \Q$ exactly corresponds to irreducible fractions with odd denominator and parity is given by the parity of the numerator. For instance, the first Collatz steps of $-\frac{4}{23}$ are: $(-\frac{4}{23},-\frac{2}{23},-\frac{1}{23},\frac{10}{23},\frac{5}{23},\,\dots)$. It reaches the cycle $(\frac{5}{23},\frac{19}{23},\frac{40}{23},\frac{20}{23},\frac{10}{23},\frac{5}{23}$ \dots).\\} (known as Lagarias' Periodicity Conjecture~\cite{10.2307/2322189}). When running the $3x+1$ FST on infinite binary inputs, one can show that it is computing the $3x+1$ function on 2-adic integers (see\if 1\mode~\cite{Collatz2arxiv},\fi\ Appendix~B). Hence, Lemma~\ref{lem:rows} and the \cqca framework in general, can be generalised for working with infinite binary inputs and the Collatz process in $\Z_2$.
\else
When running the $3x+1$ FST on infinite binary inputs, one can show that it is computing the $3x+1$ function on 2-adic integers (see\if 1\mode~\cite{Collatz2arxiv},\fi\ Remark~\ref{rk:z2}, Appendix~B). Hence, Lemma~\ref{lem:rows} and the \cqca framework in general, can be generalised for working with infinite binary inputs and the Collatz process in $\Z_2$.
\fi
\end{remark}

Intuitively, the following lemma says that columns of the \cqca\ encode all Collatz iterates, even and odd, in ternary, as in Figures~\ref{fig:model}(e) and~\ref{fig:colinput}(a). 

\begin{restatable}[Columns simulate Collatz in base $3'$]{lemma}{lemcols}\label{lem:cols}
\normalfont
Let $z \in\N^+$ and let $\alpha\in\{0,1,2\}^* \setminus \{ 0\}^*$ be the standard base 3 representation of $z$, 
and, by Lemma~\ref{lem:lc}, let $c_\infty[\alpha]\in S^\W$ be the \cqca limit configuration on input~$\alpha$. Then, 
the vertical column $x_0 < 0$ 
in $c_\infty[\alpha]$ gives 
the base~$3'$ representation of $T^{|x_0|}(z)$ (as  the defined cells down to, and excluding, the southmost cell with sum bit $0$, Definition~\ref{def:read}).
\end{restatable}
\begin{proof}
Note that it is enough to show that column $x_0 = -1$ of $c_\infty[\alpha]$ gives the base $3'$ representation of $T(z)$, with $z=\itr{\alpha}$, and then inductively apply the argument to all columns $x < -1$ to get the result. By construction of $c_0[\alpha]$ (Definition~\ref{def:ica}) and because all bits are final (Lemma~\ref{lem:lc}) we have that the sum bit of $c_\infty[\alpha](-1,|\alpha|)$ is $0$ (more generally, for all $x<0$ we have the sum bit of $c_\infty[\alpha](x,|\alpha|)$ is 0). From there, the local \cqca rule (Figure~\ref{fig:model}) is applied to each position $(-1,y)$ with $1 \leq y \leq |\alpha|$. This application of the rule corresponds exactly to running the base $3'$ $x/2$ FST (Figure~\ref{fig:fst}(b) and (d)) by reading the input on the base $3'$ symbols of cells of column $x=0$ (both sum and carry bits of these cells are defined in $c_0[\alpha]$), writing the base~$3'$ output on the cells of column $x=-1$ starting from initial FST state 0 (corresponding to the sum bit of $c_\infty[\alpha](-1,|\alpha|)$ being $0$). In that interpretation, the sum bit output to the south of a cell by the local \cqca rule corresponds to the new FST state after reading the east base $3'$ symbol, hence the sum bit $s$ of cell $(-1,0)$ corresponds to the state of the $x/2$ FST after reading all base $3'$ symbols of $\ittp{\alpha}$ (the least significant digit is at position $(0,1)$). Two cases, with $z= \itr{\alpha}$:
\begin{enumerate}
\item If $z \equiv 0 \mod 2$, \if 0\mode by Lemma~\ref{lem:fstx2},\else by Lemma~\arxivLemTwoX in the full version of this paper \cite{Collatz2arxiv},\fi \ we have that the final state of the $x/2$ FST after reading $\ittp{\alpha}$ is $0$ and that the output word $\alpha' \in \{0,\bar 0,1,\bar 1\}$ is such that $\itp{\alpha'} = \itr{\alpha}/2 = z/2$. Hence, we deduce that column $x=-1$, from $y=|\alpha|$ down to $y=1$, gives the base $3'$ representation of $T(z) = z/2$ which is what we wanted. 

\item If $z  \equiv 1 \mod 2$, \if 0\mode by Lemma~\ref{lem:fstx2},\else by Lemma~\arxivLemTwoX of \cite{Collatz2arxiv},\fi \ we have that the final state of the $x/2$ FST after reading $\ittp{\alpha}$ is $1$ and that the output word $\gamma \in \{0,\bar 0,1,\bar 1\}$, which is written on column $x=1$, $y=|\alpha|$ down to $y=1$, is such that $\itp{\gamma} = (\itr{\alpha}-1)/2 = (z-1)/2$. Furthermore, for $e=(0,0)\in \W$, the sum bit of $c_\infty[\alpha](e+\We)$, which corresponds to the final state of the $x/2$ FST, is $s=1$ and $c_0[\alpha](e) = (\bot,\bot)$. Hence, the non-local rule will be applied at position $e$ giving $c_0[\infty](e) = (0,1) = (s_e,c_e)$.
Then, at the following \cqca step, since $s_e + c_e + s =  0+1+1\geq 2$ we get a carry bit of 1 at (e+\We), i.e. $c_0[\infty](e+\We) = (1,1)$. It means that on column $x=-1$, with cell at position $(-1,0)=e+\We$ we add the base $3'$ symbol $\bar{1}$ on the least significant side of $\gamma$ (word $\gamma$ was output by the FST to the north of position $(-1,0)$). Hence, column $x=-1$, from row $y=|\alpha|$ down to $y=0$, interprets as: $3\cdot\itp{\gamma}+2 = 3\frac{z-1}{2} + 2 = \frac{3z+1}{2} = T(z)$. Which is what we wanted.
\end{enumerate}
Hence we get that column $x=-1$ gives the base $3'$ representation of $T(z)$. From the above points, it is immediate that the  cell directly below the base $3'$ expression of $T(z)$ on column $x=-1$ has sum bit equal to $0$ and that all cells below are undefined, giving the end of the Lemma statement. Note that it requires at most $|\alpha|+2$ simulation steps for the \cqca to compute that representation (at most two extra cells to the south are used).\end{proof}

We now prove our main result: the natural number written in base $3'$ on a column is converted to base 2 on the row directly south-west to it, see Figure~\ref{fig:cqca_idea}. 

\begin{theorem}[Base 3-to-2 conversion]\label{th:base_conversion}
\normalfont
Let $\alpha\in\{0,1,2\}^*$ be a finite ternary string. By Lemma~\ref{lem:lc} let $c_\infty[\alpha]\in S^\W$ be  the \cqca limit configuration on input~$\alpha$. 

Then, in $c_\infty[\alpha]$, any position $e=(x_0,y_0)\in\W$ such that both cells $e+\No$ and $e+\We$ are defined, is called an \emph{anchor cell} and has the \textit{base conversion property}:
there exists $z\in\N$ such that
the defined cells on column $x_0$ strictly to the north of $e$ give (Definition~\ref{def:read}) the base $3'$ representation of $z$ 
 and the cells on row $y=y_0$ strictly to the west of $e$ give the base $2$ representation of $z$.

 \end{theorem}
 
\begin{proof}
A direct consequence of the proof of Lemma~\ref{lem:cols} is that, in $c_\infty[\alpha]$ a defined sum bit $s\in\{0,1\}$ at position $(x_1,y_0)$ with $x_1 < 0$ and $y_0 < |\alpha|$ gives the state of the $x/2$ FST (Figure~\ref{fig:fst}(b) and (d)) immediately after it reads all base $3'$ symbols from row $y=|\alpha|$ to $y=y_0+1$ on column with x-coordinate $x_1+1$ of $c_\infty[\alpha]$.
As $s$ is the final FST state, \if 0\mode by Lemma~\ref{lem:fstx2},\else by Lemma~\arxivLemTwoX~\cite{Collatz2arxiv},\fi \ we get that $s$ is 0 if that base $3'$ represented number was even; otherwise (if odd) $s$ is $1$.

We also know that the output of the FST, i.e. symbols strictly to the north of $(x_1,y_0)$ represent $\lfloor x/2 \rfloor $. Hence, one base conversion step was performed: $x\mod 2$ is written in the sum bit $s$ at position $(x_1,y_0)$ and $\lfloor x/2 \rfloor$ is computed to the north of it. By induction, all the other base conversion steps are also performed to the west of $(x_1,y_0)$ and we get the result for the anchor cell $(x_0,y_0)$ with $x_0 = x_1 + 1$.
\end{proof}

\begin{remark} 
Figure~\ref{fig:influence}(a) presents several instances of Theorem~\ref{th:base_conversion}. Note that  position $(0,0)$ is always an anchor cell in $c_0[\alpha]$ meaning that the \cqca converts $\ittp{\alpha}$ from base $3'$ to base $2$. Hence, the \cqca can effectively convert any base $3'$ input to base $2$. Figure~\ref{fig:influence}(b) presents the details of such a base conversion and shows that the \cqca base conversion algorithm is rather natural: at each step the parity of the input is computed, then the input is divided by $2$. Also, the algorithm is efficient: for $x\in\N$, it can be shown that $\lceil \log_2(x)\rceil + \lceil \log_3(x)\rceil$ \cqca steps are sufficient to convert $x$ from base $3$ to base 2.
\end{remark}

\begin{remark}
Theorem~\ref{th:base_conversion} also implies that limit configurations of row inputs are very similar to limit configurations of column inputs. Indeed, for $z\in\N$, with base 2 
representation $w\in\{0,1\}^*$ and base 3  representation  $\alpha\in\{0,1,2\}^*$, we have: for all $x \leq 0$ and $y \leq 0$, $c_\infty[w](x,y) = c_\infty[\alpha](x,y)$. 
Thus, for all $x\leq -|w|$ and $y < 0$, columns of $c_\infty[\omega]$ iterate the Collatz process in ternary and the base conversion property holds for any anchor cell. Hence, the base conversion property naturally appears in the \cqca, for both row and column inputs.
\end{remark}

\begin{remark}
Theorem~\ref{th:base_conversion} implies that cells with state $(0,1)$ because of the non-local rule (red carries in Figure~\ref{fig:model} and \ref{fig:influence}) have an interesting interpretation column-wise: they implement the operation $3x+1$ in ternary. Indeed, in base $3$ the operation $3x+1$ is trivial: it consists of appending 1 (represented here by $\bar 0$ via Definition~\ref{def:btp}) to the base $3'$ representation of $x$. 
Thus, the \cqca ``stacks'' trivial ternary steps ($3x+1$) on the same column, and trivial binary steps ($x/2$, i.e.\ a shift in binary) on the same row (Remark~\ref{rk:row1}).
\end{remark}

\begin{corollary}[Parity checking in Collatz]\label{lem:parityCheck}
\normalfont
Let $\alpha \in \{0,1,2\}^*$ 
and, by Lemma~\ref{lem:lc}, let $c_\infty[\alpha]\in S^\W$ be the limit configuration of the \cqca\ on input~$\alpha$ (written in base $3'$ on column $x=0$). 
For any anchor cell (Theorem~\ref{th:base_conversion}) at postion $e\in\W$, 
let $s \in \{0,1\}$ be the sum bit of the cell at  $e+\We$. Then, $s$ is the parity of the number written in base $3'$ on the column at $e+\No$ and going to the north (Definition~\ref{def:read}): this number is even iff $s=0$ and odd iff $s=1$.
\end{corollary}
\begin{proof}
Immediately implied by the proof of Theorem~\ref{th:base_conversion}: the sum bit $s$ at position $e+\We$ gives the state in which the $x/2$ FST was after reading base $3'$ symbols on the column to the north of $e$. \if 0\mode By Lemma~\ref{lem:fstx2},\else By Lemma~\arxivLemTwoX in the full version of this paper \cite{Collatz2arxiv},\fi \ the bit $s$ is the parity of the number given by that column.
\end{proof}

\begin{example} 
Figure~\ref{fig:cqca_idea}(b) outlines instances of Corollary~\ref{lem:parityCheck}.
\end{example}

\subsection{Computational complexity of \cqca prediction}\label{sec:complexity}

In this section we leverage our previous results to make statements about the computational complexity of predicting the \cqca. 
 
\begin{definition}[Bounded \cqca prediction problem]\label{def:bcpp}
\normalfont
Given (a)
any ternary input $\alpha$,  of length $n\in\N$ with  resulting \cqca  limit configuration  $c_\infty[\alpha]$,   
and (b)~any $(x,y) \in \Z^2$, where $\max(|x|,|y|) = O( n )$,  what is the state $c_\infty[\alpha](x,y)$?
\end{definition}

\begin{remark}
The version of this prediction problem, where the question is to predict the state $c_\infty[\alpha](x,y)$ for any $(x,y)\in\W$, is at least as hard as the Collatz conjecture which in \cqca terms states that the $(1,2,1,2,\dots)$ ``glider'' will eventually occur, cf. blue outlined cells in Figure~\ref{fig:model}(d).
\end{remark}

It is straightforward to see that the bounded \cqca prediction problem is in the complexity class P, we can also give a lower bound in terms of AC$^0$ which is the class of problems solved by uniform\footnote{Here, uniform has the meaning used in Boolean circuit complexity:  that members of the circuit family for an infinite problem (set of words) are produced by a suitably simple algorithm~\cite{hesse2002uniform}. The class AC$^0$ is of interest as it is ``simple'' enough to be strictly contained in P, that is, AC$^0 \subsetneq$ P~\cite{moore2011nature}. This enables us to give lower bounds to the computational complexity, or inherent difficulty,  of some problems, such as the bounded \cqca prediction problem. } polynomial size, constant depth Boolean circuits with arbitrary gate fanin~\cite{moore2011nature}:

\begin{restatable}{theorem}{boundedPredAC}
\label{thm:not in AC0} 
\normalfont
The bounded \cqca prediction problem is in P and not in AC$^0$. 
\end{restatable}
\begin{proof}
To see that the problem is in P, simply  encode the initial configuration as input to a two-tape Turing machine and, in $O(n^2)$ time,  run the simulation forward until we have filled the plane up to distance $n$ from the input.

To show the problem is outside AC$^0$, 
 let $v \in\{0,1\}^*$, and let $\alpha \in \{0,1,2\}^*$ be the base $3$ word such that $\alpha = v$. 
 Let $c_\infty[\alpha]$ be the limit configuration of the \cqca on (column) input $\alpha$ written in base $3'$ (as usual) on column $x=0$. 
 Since there are no 2-symbols in $v$, by Definition~\ref{def:btp}, the base $3'$ encoding of $v$ maps 0 to 0 and 1 to $\bar 0$, an encoding straightforward to represent in binary, and straightforward to compute in AC$^0$.
   Let $b$ be the sum bit of $c_\infty[\alpha](-1,0)$,  i.e. $(-1,0)+\No+\Ea$ is  the first symbol of $\alpha$, hence (well) within the bound $n=O(|\alpha|)$ set by Definition~\ref{def:bcpp}.

Deciding whether a natural number is odd or even is equivalent to checking the parity of its number of $1$s written in base $3$. In base $3'$, this translates to checking the parity of the total number of sum and carry bits.
 By Corollary~\ref{lem:parityCheck}, $b$ is the parity of the natural number represented by $\alpha$, equivalently, $b$ is the parity of the number of 1s in $v$.
    Hence the \cqca solves the problem  
     $\texttt{PARITY} = \{ v \in \{0,1\}^* \, \mid v \text{ has an odd number of } 1\text{s}\}$
     on the input~$v$, with the result placed at  distance $2 < |\alpha|$ from the input word $\alpha$.  
    Since $\texttt{PARITY}$ sits outside the complexity class AC$^0$~\cite{moore2011nature}, the 
    Bounded \cqca prediction problem is not in   AC$^0$. 
\end{proof}

\begin{remark}
The proof of the previous theorem shows how to use the \cqca to solve $\texttt{PARITY}$. In fact, we can say something stronger:  in any \cqca configuration, each sum bit with defined $\No + \Ea$ neighbour solves an instance of the $\texttt{PARITY}$ problem. See Figure~\ref{fig:cqca_idea}. 
\end{remark}

\newcommand{\Arbcpp}{Upper $z$-rectangle prediction problem\xspace}
\newcommand{\arbcpp}{upper $z$-rectangle prediction problem\xspace}
\begin{definition}[\Arbcpp]\label{def:rcpp}
\normalfont
Let $\alpha \in \{0,1,2\}^\ast$, $z = \itr{\alpha}\in\N$ and let $n = \lceil \log_2 z\rceil \in \N$.
Let $c[\alpha]$ be the associated \cqca initial configuration (Definition~\ref{def:ica}) and $R = \{ (x,y) \mid  -n \leq x \leq 0, 0 \leq y \leq |\alpha|  \} \subsetneq \Z^2$.
The \arbcpp asks: What are the states, in the limit configuration $c_\infty[\alpha]$, of all cells with positions  $(x,y) \in R$. 
\end{definition}

\begin{example}
Figure~\ref{fig:influence}(c) gives a schematic representation of $R$, the upper $z$-rectangle. Figure~\ref{fig:influence}(d) and Figure~\ref{fig:smallC} each give and instance of $R$, respectively for $z=\ibin{101111011011}=3035$ and $z=\ibin{11110010}=242$.
\end{example}

NC$^1$ is the class of problems solved by uniform polynomial size, logarithmic depth (in input length) Boolean circuits with gate fanin  $\leq 2$~\cite{hesse2002uniform}.
The proof of the following theorem intuitively comes from the fact that  predicting the entire upper $z$-rectangle amounts to running $\lceil \log_3 z\rceil$ base conversions in parallel\if 0\mode . \else 
 \ (the proof is in the full version of this paper~\cite{Collatz2arxiv}, Theorem~\arxivThComplexity).\fi

\begin{restatable}{theorem}{ncone}
\label{thm:nc1}
\normalfont
The \arbcpp is in NC$^1$, and is  not in AC$^0$. 
\end{restatable}
\if 0\mode
\begin{proof}
 Not being in AC$^0$ follows from the same proof as that of Theorem~\ref{thm:not in AC0}; in particular, by Definition~\ref{def:rcpp}, the position $(-1,0)\in\W$ used in that proof is in $R$ (Definition~\ref{def:rcpp}). 

 Hesse, Allender,  Barrington~\cite{hesse2002uniform} 
 have shown that conversion from base~3 to~2 is in uniform NC$^1$. 
 By Theorem~\ref{th:base_conversion}, 
 each row $y$ of the rectangle $R$ contains exactly the base 2 representation of the base 3 number that is written on column $x=0$, with its least significant digit at position $(0,y+1)$. 
 Hence $|\alpha|$ parallel applications of a NC$^1$ base conversion circuit (of polynomial size and logarithmic depth, in $|\alpha|$) compute all $|\alpha|$ rows of $c_\infty[\alpha]$ that are contained in the rectangular region~$R$, which in turn gives a uniform\footnote{The proof of uniformity follows from the simplicity of the definition of $R$ --- each of the $|\alpha|$ outputs map to a single row of $R$-- and would use somewhat technical notions of DLOGTIME, or DLOGTIME-uniformity-AC$^0$, uniformity for NC$^1$~\cite{barrington1990uniformity}.
  Details are omitted for brevity.} NC$^1$ circuit. 
\end{proof}
\fi

\begin{op}[Lower $z$-rectangle prediction problem]\label{op:lower rect}
What is the complexity of filling out the lower $z$-rectangle? I.e. the rectangle defined by $M = \{ (x,y) \mid  -n \leq x \leq 0,  -n \leq y \leq 0 \} \subsetneq \Z^2$ (same notation as Definition~\ref{def:rcpp} and shown on Figure~\ref{fig:cqca_idea}(c)). We know it is in P, and that it is not  in AC$^0$ (Theorem~\ref{thm:not in AC0}). Can we get matching lower and upper bounds? 
\end{op}

\begin{figure}[t]
\centering
  \begin{minipage}[c]{0.75\textwidth}
    \caption{
    {A small positive cycle, almost. The  rightmost magenta column encodes
    $z = \itp{\bar 1 \bar 1 \bar 1 \bar 1 \bar 1} = 242$ in base~$3'$.
    By Lemma~\ref{lem:cols}, the leftmost magenta column is $T^{\lceil \log_2(x) \rceil}(x) = \itp{\bar 1 \bar 1 1 \bar 1 \bar 1} = \itr{22122} = 233$. The numbers $242$ and $233$ differ in only one trit (in blue): they almost generate a small positive Collatz cycle (Definition~\ref{def:sc}). 
    Theorem~\ref{th:base_conversion} tells us that the region above the magenta row is characterised by base conversion, what about the region below? 
     } 
    } \label{fig:smallC}
  \end{minipage}\hspace{0.5ex}
    \begin{minipage}[c]{0.22\textwidth}
\if 1\mode \vspace{-3ex}\fi\includegraphics[width=1\textwidth]{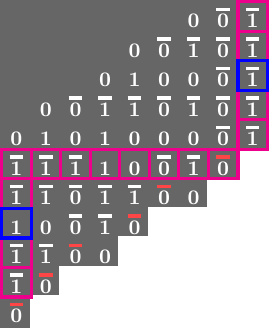}
  \end{minipage}
\if 1\mode  \vspace{-5ex}\fi
\end{figure}

These prediction problems are closely related to what we call \emph{small positive Collatz cycles}. Indeed, if predicting $M$ (Open Problem~\ref{op:lower rect}) was easy, one could hope to easily answer whether there exists $x > 2$ such that $x = T^{\leq \lceil \log_2(x) \rceil}(x)$:
\begin{definition}[Small positive Collatz cycles]\label{def:sc}
\normalfont
Let $x\in\N^+$. We say that $x$ generates a small positive Collatz cycle if $x = T^m(x)$ with $0<m\leq {\lceil \log_2(x) \rceil}$.
\end{definition}

\begin{op}
There are no small positive Collatz cycles besides $(2,1,2,1,\ldots)$.
\end{op}

\begin{example}
 Figure~\ref{fig:smallC} shows that answering the question about small positive Collatz cycle seems challenging. Indeed, the Figure illustrates that for $x = \itp{\bar 1 \bar 1 \bar 1 \bar 1 \bar 1} = \itr{22222} = 242$, we have $T^{\lceil \log_2(x) \rceil}(x) = \itr{22122} = 233$. The two numbers differ only in one trit, it is \emph{almost} a small positive Collatz cycle. 
\end{example}

\begin{remark}
One can also reformulate the small positive Collatz cycles problem by saying that, although the Collatz process is able to compute base $3'$-to-2 conversion, it can never compute base 2-to-$3'$ in $\lceil \log_2(x) \rceil$ \cqca steps (except for $x= 2$). Note that the \cqca would have to produce base $3$ trits, from most to least significant trit which seems very hard to in the $\lceil \log_2(x) \rceil$ time constraint.
\end{remark}

\if 1\mode
\section{Discussion:~Structural~implications~on~Collatz~sequences}\label{sec:discussion}
\else
\subsection{Discussion:~Structural~implications~on~Collatz~sequences}\label{sec:discussion}
\fi

Figure~\ref{fig:influence} summarises some strong consequences of our results  
on the structure of Collatz sequences. 
First, since columns are iterating the Collatz process in base~$3$ (Lemma~\ref{lem:cols}), the base conversion theorem implies that any
$z\in\N$ is giving some rather specific constraints on the next $\lceil \log_2 z \rceil$ iterations of the Collatz process.
Specifically, the upper $z$-rectangle, Figure~\ref{fig:influence}(b), which corresponds to the diagram of the conversion of $z$ from base $3'$ to base $2$, is constraining, on average, half of the trits (base 3 digits) of any iteration $T^{\leq \lceil \log_2 z \rceil(z)}(z)$, Figure~\ref{fig:influence}(d). 

Second, Theorem~\ref{thm:nc1} tells us that this upper $z$-rectangle is easy to predict, in the sense that the entire region can be computed in NC$^1$. The computational complexity of lower $z$-rectangle prediction remains open.

Third, Figure~\ref{fig:influence}(b) illustrates how the parity checking result of Corollary~\ref{lem:parityCheck} places constraints on future iterates. A sum bit at {\em any} position  $e\in\W$ of a configuration $c_\infty[\alpha]$ is constrained to be the parity of the number of 1s  in the entire column (both sums and carries) whose base is at the north-east of $e$.

We should note that these phenomena are occurring everywhere, at each Collatz iterate.  
Hence, although patterns have been notoriously hard to identify in the Collatz process, 
our results give a new lens which reveals some detail.

\if 1\mode
\vspace{\baselineskip}

\noindent {\bf Acknowledgement.} 
Sincere thanks to Olivier Rozier for fruitful interactions, and to the anonymous reviewers for helpful comments. 
\fi

\if 0\mode

\section{Constructing rational cycles in the Collatz process}\label{sec:cycles}

Here, we study rational cycles in the Collatz process by exploiting \emph{parity vectors}, a notion central to the study of the Collatz process \cite{Terras1976,10.2307/2322189,wirsching1998the,Monks2004,Rozier2019ParitySO}.
We show that parity vectors can naturally be exploited in the \emph{reversed} \cqca\ (\rcqca). From there we exhibit a construction to generate the 2-adic expansion of any rational 2-adic integer involved in a Collatz cycle. We reformulate the cyclic Collatz conjecture in $\N$ and $\Z$ as a simple prediction problem in the \rcqca. Finally, this Section leaves open the question of how the base conversion abilities of the \cqca (Theorem~\ref{th:base_conversion}) generalise to the conversion of any element of $\Z_2 \cap \Z_3 \cap \Q$ from $\Z_2$ (i.e.  identified with infinite binary expansion) to $\Z_3$ (i.e. identified with infinite ternary expansion).

\subsection{Overview of the Collatz process in $\Z_2$} As mentionned in Section~\ref{sec:2dcqca}, the most general version \cite{10.2307/2322189} of the Collatz process is defined on the uncountable ring $\Z_2$. Elements of $\Z_2$ are identified with the set of one-way infinite binary strings, hence for notational
convenience we let $\Z_2$ denote either the set 2-adic numbers or the set of one-way infinite binary strings, depending on the context. Following the convention for finite binary strings, in this paper we represent elements of $\Z_2$ with their least significant bit to the right,\footnote{This is in contrast with the mathematic literature where 2-adic integers are usually represented with their least significant bit to the left, e.g. $101010\ldots \in \Z_2$.} e.g.\ we write $\ldots010101\in\Z_2$. The operations $+$ and $\times$ in $\Z_2$ are defined bitwise as the natural extension of standard finite binary addition and multiplication~\cite{caruso:hal-01444183}. For instance, we can easily interpret $z= (01)^\infty$ as being the number $-\frac{1}{3}$ since $3z = 2z + z = (10)^\infty + (01)^\infty = (1)^\infty$, and so $3z + 1 = 0$. Note that the set of natural numbers $\N$ is identified with $\{(0)^\infty w \mid w \in \{0,1\}^* \} \subset \Z_2$ and the set of strictly negative integers $\Z \setminus \N$ with $\{(1)^\infty w \mid w \in \{0,1\}^* \}\subset \Z_2$ (matching the standard ``two's complement'' representation of signed integers). Finally, $\Q \cap \Z_2$ is the set of $2$-adic integers with eventually periodic expansion \cite{Conrad}. From the point of view of $\Q$, the set $\Q \cap \Z_2$ is the set of irreducible fractions with \emph{odd} denominator\footnote{We remark, for instance, that it is impossible to represent a number such as $x=\frac{1}{2}$ in $\Z_2$. Indeed we cannot have $2x=1$ since multiplying by $2$ appends a $0$ to the least significant side of $x$.}, e.g. $-\frac{4}{23} \in \Q \cap \Z_2$. With all this context in mind, it becomes clear how the Collatz process generalises to $\Z_2$: if $z\in\Z_2$ is even, i.e.\ the least significant bit is~$0$, we remove that bit and if $z\in\Z_2$ is odd, i.e.\ the least significant bit is~$1$, we run it through the $3x+1$ FST whose logic directly generalises to $\Z_2$ (Remark~\ref{rk:z2}). The generalised version of the Collatz conjecture to $\Z_2$ is known as Lagarias' periodicity conjecture and can be stated as follows: ``Any rational 2-adic $z\in\Z_2 \cap \Q$ eventually reaches a cycle under action of the Collatz process''. Finally, it is known that Collatz cycles in $\Z_2$ contain only rational numbers implying that  irrational 2-adic integers (i.e.\ aperiodic binary expansion) never reach any cycle~\cite{wirsching1998the}.

\subsection{Parity vectors and the reverse \cqca}\label{sec:pv}
\subsubsection{Parity vectors}\label{sec:pvpv}
\begin{figure}[h!]

\begin{subfigure}[b]{0.475\textwidth}
\centering
\includegraphics[scale=0.8]{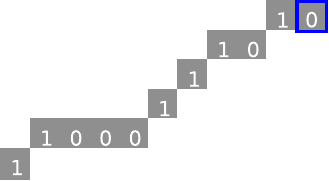} 
\caption{Initial parity vector configuration in the \rcqca.}
\end{subfigure}
\hfill
\begin{subfigure}[b]{0.475\textwidth}
\centering
\includegraphics[scale=0.8]{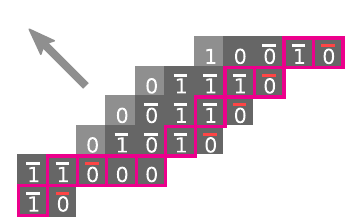} 
\caption{Evolution after several \rcqca steps: the arrow indicates the north-west evolution direction.}
\end{subfigure}
\vskip\baselineskip 
\centering
\begin{subfigure}[b]{0.45\textwidth}
\centering
\includegraphics[scale=0.8]{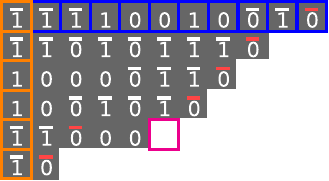} 
\caption{Evolution after complete reconstruction of row $y=0$.}
\end{subfigure}

\caption{Running the \rcqca from a finite parity vector. Same colour codes as Figure~\ref{fig:model}. (a) Initial \rcqca configuration for the parity vector
$p=(0,1,0,1,1,1,0,0,0,1,1)\in\Pa$ (Definition~\ref{def:initp}). The  cell at position $(0,0)$ is outlined in blue. (b) Configuration after several \rcqca steps. Evolution goes north-west. Because of the non-local rule, some cells have appeared to the east of the initial parity vector (magenta). (c) Because the parity vector is finite, the computation area is bounded and will eventually terminate when all cells on row $y=0$ and $-|p|\leq x \leq 0$ are computed (in blue). By Lemma~\ref{lem:pv}, row $y=0$ gives the binary encoding of the smallest  $z\in\N$ such that $|p|$ is the parity vector of $z$'s first $|p|$ Collatz steps. Furthermore, again by Lemma~\ref{lem:pv}, the leftmost column (orange) gives the ternary expansion of $T^{|p|}(z)$. Indeed, here we have $z = \ibin{11110010010} = 1938$  (as read on the blue line) and $T^{11}(z) = 692 = \itp{\bar 1 \bar 1 1 1 \bar 1 \bar 1} = \itr{221122}$ (as read on the orange column). Note that, in the forward \cqca, the cell outlined in magenta would be in state $(0,0)$. But, in the backward direction, the rule can never apply to this cell hence, it is left undefined. This is related to the fact that we study parity vectors of the map $T$ where $T_1(x) = (3x+1)/2$ and not $T_1(x)=3x+1$.}\label{fig:border}
\end{figure}

For $z\in\Z_2$, the parity vector associated to the Collatz trajectory $\mathcal{O}(z) = (z,T(z),T^2(z),\dots)$ is the sequence $p = (z \mod 2, T(z) \mod 2, T^2(z) \mod 2,\dots)$. We write $\Pa = \cup_{n\in\N}\{0,1\}^n$ as the set of finite parity vectors and we write for instance $p = (p_0,p_1,p_2,p_3) = (1,0,0,1) \in \Pa$ which is a valid parity vector for the finite Collatz trajectory $(13,20,10,5)$. We write $|p|$ for the size of the parity vector $p$. Parity vectors have an important property: they are all \emph{feasible}, meaning that for any $p\in\Pa$ there is $z\in\N$ such that $p$ is the parity vector of the first $|p|$ Collatz steps of $z$ \cite{wirsching1998the}. Furthermore, if $z_1,z_2\in\N$ share the same last $n$ bits in their binary representation (i.e. $z_1 \equiv z_2 \mod 2^n$) then the parity vector of their first $n$ Collatz steps are the same\footnote{There is in fact a mysterious relationship between the $\lceil \log(z) \rceil$ bits of $z\in\N$ and the elements in the parity vector of its first $\lceil \log(z) \rceil$ Collatz step \cite{1907.00775_RP_2020_sub}.}. More generally, considering the infinite parity vector associated to the Collatz sequence of any 2-adic integer gives the map $Q: \Z_2 \to \{0,1\}^\N$ which maps $z$ to its parity vector. The map $Q$ is one-to-one and onto and has other remarkable properties \cite{10.2307/2322189,Rozier2019ParitySO}. In that setting, Lagarias' periodicity conjecture reformulates as follows: if $z\in\Z_2$ is eventually periodic (i.e. $z\in\Q\cap\Z_2$) then $Q(z)$ is also eventually periodic\footnote{The converse is known: if $Q(z)$ is eventually periodic then $z$ also is \cite{10.2307/2322189,Rozier2019ParitySO}.} \cite{10.2307/2322189}. 

\subsubsection{The reverse \cqca} The reverse \cqca (\rcqca) shares the same definition as the \cqca except for its local rule which, intuitively,  inverses the action of the \cqca. 
More precisely, by inverse we mean that, when considering a cell at position $e\in\W$, if both the bit and carry of $e+\Ea$ are defined and the bit of $e+\So$ is defined, then there is only one bit and carry choice for $e$ which match with the local rule of the \cqca. For instance, if $c(e+\Ea) = (\color{blue} 1\color{black}, \color{red}{1}\color{black})$ and $c(e+\So) = (\color{black}{0}\color{black},\bot)$ then we must have $c(e) = (\color{violet}{0}\color{black},\color{violet}{1}\color{black})$ since $0 + \color{blue} 1\color{black} + \color{red}{1}\color{black} = \color{violet}{0}\color{black} \mod 2$ and $0 + \color{blue} 1\color{black} + \color{red}{1}\color{black} \geq 2$. Note that the logic of the \rcqca corresponds, horizontally, to running the inverse of the $3x+1$ FST and, vertically (from bottom to top), to running the $x/2$ FST of which arrows have been inverted (Remark~\ref{rk:inv}, Appendix~\ref{app:fst}).
\ \\

Inputting a parity vector $p$ to the \rcqca is done as follows: 

\begin{definition}[Initial configuration \ensuremath{c_0[p]}]\label{def:initp}
\normalfont
Let $p\in\Pa$ be a finite parity vector. Then $c_0[p]\in S^\W$, the initial configuration associated to $p$ is defined by the following process: start at position $e = (0,0)$ and set $c_0[p](e) = (p_0,\bot)$. If $p_0 = 0$ move to position $e'=e+\We$ else move to position $e'=e+\So +\We$, repeat with the next entry in the parity vector or stop if the end is reached. Set $c_0[p](e) = (\bot,\bot)$ at every other position $e$.  
\end{definition}

\begin{remark}
Similarly to Lemma~\ref{lem:lc}, and with a similar proof, the limit configuration $c_\infty[p]$ associated to input parity vector $p$ is well defined.
\end{remark}

\begin{example}
Figure~\ref{fig:border}(a) and Figure~\ref{fig:border}(c) respectively show $c_0[p]$ (position $(0,0)$ in blue) and $c_\infty[p]$ for  $p=(0,1,0,1,1,1,0,0,0,1,1)$. Figure~\ref{fig:border}(b) shows how the local rule of the \rcqca is applied.
\end{example}

We have the following general result:

\begin{lemma}\label{lem:pv}
\normalfont
Let $p\in\Pa$ be a finite parity vector and $c_\infty[p]$ be the limit configuration of the \rcqca when $p$ is given as initial input (Definition~\ref{def:initp}). Then, on row $y=0$ of $c_\infty[p]$ we read the binary representation of the smallest natural number $z\in\N$ such that $p$ is the parity vector of the first $|p|$ Collatz steps of $z$. Also, on row $x=-|p|$ we read the base $3'$ representation of $T^{|p|}(z)$.
\end{lemma}
\begin{proof}
This is immediate by reversibility of the \cqca rule. Running the \cqca from the row $y=0$ constructed by the \rcqca will lead to the same choices for each cell, and produce the same parity vector. It is known that for a parity vector of size $n$, only one $z < 2^n$ is such that it follows that parity vector during its first $n$ Collatz steps \cite{wirsching1998the}. Hence, row $y=0$ gives the binary representation of the smallest natural number $z\in\N$ such that $p$ is the parity vector of its first $|p|$ Collatz steps. The fact that column with x-coordinate $x=-|p|$ gives the base $3'$ representation of $T^{|p|}(z)$ follows from Theorem~\ref{th:base_conversion} applied at anchor cell at position $(-|p|,-|p|_1 -1)$ with $|p|_1$ the number of $1$s in $p$.
\end{proof}

\begin{example}
Figure~\ref{fig:border}(c) gives an example of Lemma~\ref{lem:pv} for $p=(0,1,0,1,1,1,0,0,0,1,1)$. On row $y=0$, we read, $z = \ibin{11110010010} = 1938$. It is the smallest integer for which the first $|p|=11$ Collatz iterates have the parity prescribed by $p=(0,1,0,1,1,1,0,0,0,1,1)$. Indeed, the first $11$ Collatz iterates of $1938$ are: $(1938, 969, 1454, 727, 1091, 1637, 2456, 1228, 614, 307, 461)$.
\end{example}

\begin{remark}
Lemma~\ref{lem:pv} can be effortlessly generalised to half-infinite parity vectors (first bit is written at $(0,0)$, and going to the south-west). In which case, for an infinite parity vector $p$, row $y=0$ gives the 2-adic expansion of $Q^{-1}(p)$ ($Q$ is defined in Section~\ref{sec:pvpv}, see~\cite{10.2307/2322189,Rozier2019ParitySO}).
\end{remark}

\begin{remark}
The cyclic Collatz conjecture in $\N$ is equivalent to: the only parity vectors $p\in\Pa$ for which $z = T^{|p|}(z)$ with $z\in\N$ are powers (in the sense of concatenation) of $(0,1)\in\Pa$ and powers of $(1,0)\in\Pa$, e.g. $(0,1,0,1,0,1)$ or $(1,0,1,0,1,0)$.
 Note that it is equivalent to saying that, in Figure~\ref{fig:border}, the leftmost (orange) column never contains the base $3'$ representation of the row $y=0$ (blue), given that the number represented by the row $y=0$ is greater than $2$.

\end{remark}

\subsection{Parity vectors and the cyclic world $\W/\!\sim_p$}
\subsubsection{Constructing $2$-adic expansions of Collatz cyclic numbers}

\begin{figure}[h!]

\begin{subfigure}[b]{0.475\textwidth}
\centering
\includegraphics[scale=0.8]{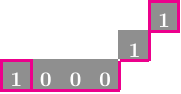} 
\caption{Initial parity vector configuration in the cyclic world $\W/\!\sim_p$. \vspace{\baselineskip}}
\end{subfigure}
\hfill
\begin{subfigure}[b]{0.475\textwidth}
\centering
\includegraphics[scale=0.68]{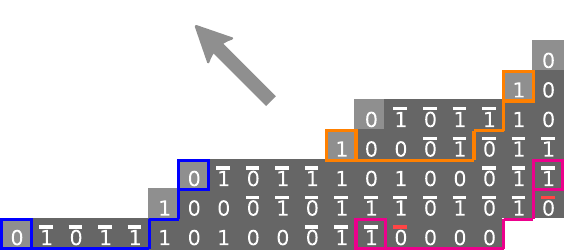} 
\caption{Evolution of the \rcqca in the cyclic world: the arrow indicates the north-west evolution direction.}
\end{subfigure}
\vskip\baselineskip 
\centering
\begin{subfigure}[b]{0.45\textwidth}
\centering
\includegraphics[scale=0.63]{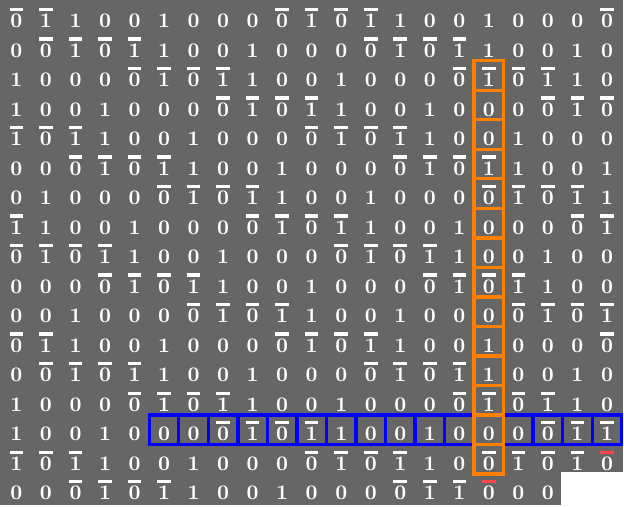} 
\caption{Part of the limit configuration in the cyclic world.}
\end{subfigure}
\hfill
\begin{subfigure}[b]{0.45\textwidth}
\centering
\includegraphics[scale=0.63]{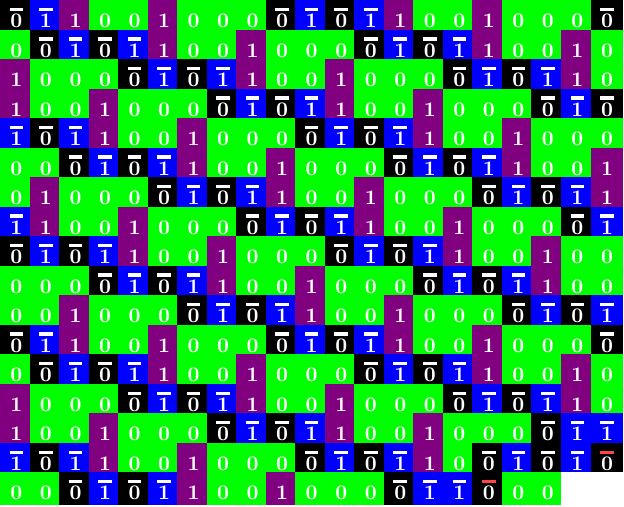} 
\caption{Same as (c) with one color per state $(0,0),(0,1),(1,0),(1,1)$.}
\end{subfigure}

\caption{\small The \rcqca in the cyclic world of parity vector $p=(1,1,0,0,0)\in\Pa$. 
Cells colours white, light grey and dark grey have the same meanings as Figure~\ref{fig:model}, except for (d). (a) Initial configuration $\tilde{c}_0[p]$ in the cyclic world $\W/\!\sim_p$ (Definition~\ref{def:initpt}). The two cells outlined in magenta are at equivalent positions in $\W/\!\sim_p$: they contain the same state. All cells underlined in magenta are in the same parity vector cut (Definition~\ref{def:pvc}). (b) After some \rcqca steps in the cyclic world. Only the quadrant $x <0$ and $y>-|p|_1$ is shown. Cells outlined with the same color are at equivalent positions: they must contain the same state. Cells underlined with the same color are in the same parity vector cut (Definition~\ref{def:pvc}). (c)~Quadrant $x <0$ and  $y>-|p|_1$ of the limit configuration $\tilde{c}_\infty[p]$ (Definition~\ref{def:initpt}). On row $y=0$ (in blue, the outlined cell stop at the end of the first period of the expansion), by Theorem~\ref{th:cycleconstruct}, we read the 2-adic expansion of $Q^{-1}(p^\infty)$, the Collatz cyclic number associated to $p=(1,1,0,0,0)\in\Pa$ (Definition~\ref{def:ccn}). In this case, we have $Q^{-1}(p^\infty)=\frac{5}{23}$ and indeed we read $(00010110010)^\infty 00011 \in \Z_2$ which is the 2-adic expansion of $\frac{5}{23}$. Conjecture~\ref{cj:bc} states that the orange column (the outlined cell stop at the end of the first period of the expansion) gives the $3$-adic expansion of $z$, it is the case because the $3$-adic expansion $(20021001011)^\infty 201$, given by reading the base $3'$ symbols of the orange column from south to north, is the $3$-adic expansion of $\frac{5}{23}$ (see Example~\ref{ex:cycle} which gives the formula to apply in order to transform an eventually periodic 2-adic/3-adic expansion into its associated fraction in $\Q$). (d) Same as (c) but individual colors have been displayed for states $(0,0),(0,1),(1,0),(1,1)$ in order to visualize the patterns generated by the \rcqca in the cyclic world of $p$.}\label{fig:cycle1}
\end{figure}

We say that a parity vector $p\in\Pa$ is the support of a cycle if there exists a Collatz cycle of period length $|p|$ of which $i^\text{th}$ element has the parity (i.e. least significant bit) prescribed by $p_i$. For instance $(0,1,0,1)\in\Pa$ is support to the cycle $(2,1,2,1)$ and the parity vector $(1,1,1,1)\in\Pa$ is support to the cycle $(-1,-1,-1,-1)$. It is known that for any $p\in\Pa$ there is a $z \in \Z_2 \cap \Q$ such that $p$ supports the cycle $(z,T(z),\dots,T^{|p|}(z))$ and such cycles are the only one existing in the Collatz process \cite{10.2307/2322189,wirsching1998the}. We call such $z\in \Z_2 \cap \Q$, \emph{Collatz cyclic numbers}:

\begin{definition}[Collatz cyclic numbers]\label{def:ccn}
Let $p\in\Pa$ and let $p^\infty$ denote $pp\cdots $.  
The Collatz cyclic number $z\in\Z_2 \cap \Q$ associated to $p$ is given by $z = Q^{-1}(p^\infty)$ where $Q: \Z_2 \to \{0,1\}^\N$ the one-to-one and onto map \cite{10.2307/2322189,Rozier2019ParitySO} which associates each 2-adic integer to its Collatz sequence's parity vector.
\end{definition}

We show that the \rcqca naturally yields to the construction of the 2-adic expansion of any Collatz cyclic number.
In order to provide the construction, we need to put a equivalence relation on the grid $\W$ which identifies some positions with each others:
\begin{definition}[The cyclic world \ensuremath{\W/\!\sim_p}]\label{def:initpt}
\normalfont
Let $p\in\Pa$ be a finite parity vector. Define $v_p = |p|*\We + |p|_1*\So$ with $|p|_1$ the number of $1$s in $p$. Define $\sim_p$ the following equivalence relation on $\W$: for $e,f\in \W$, $e\sim_p f$ if and only if $e-f = v_p$. We write $\W/\!\sim_p$ the quotient of $\W$ by $\sim_p$. Define $\tilde{c}_0[p] \in S^{\W/\!\sim_p}$ in the same way than $c_0[p] \in S^\W$. It is easily shown that the limit configuration  under action of the \rcqca, $\tilde{c}_\infty[p] \in S^{\W/\!\sim_p}$, is well defined.
\end{definition}

\begin{example}
Figure~\ref{fig:cycle1}(a) shows the initial configuration $\tilde{c}_0[p]$ associated to $p$ in the cyclic world $\W/\!\sim_p$. The two magenta outlined cells are at equivalent positions $e_1,e_2 \in \W$: we have $e_1 \sim_p e_2$, meaning that in $\W/\!\sim_p$ those two cells are the same. In particular they have the same state. Similary for cells outlined with the same colour in Figure~\ref{fig:cycle1}(b). Figure~\ref{fig:cycle1}(c) gives a portion of the limit configuration $\tilde{c}_\infty[p]$.
\end{example}

\begin{theorem}[Constructing Collatz cyclic numbers]\label{th:cycleconstruct}
\normalfont
Let $p\in\Pa$ be a finite parity vector and let $\tilde{c}_\infty[p] \in S^{\W/\!\sim_p}$ be the limit configuration of the \rcqca when $p$ is given as input in the cyclic world $\W/\!\sim_p$ (Definition~\ref{def:initpt}). Let $z = Q^{-1}(p^\infty) \in \Z_2 \cap \Q$ be the Collatz cyclic number associated to $p$ (Definition~\ref{def:ccn}). 

Then, the $2$-adic expansion of $z$ can be read by concatenating all the sum bits on row $y=0$.

\end{theorem}
\begin{proof}
By reversibility of the local rule of the \cqca running the forward \cqca from row $y=0$ constructed by the \rcqca will lead to the same choices at each cell. In particular that means, that the $2$-adic number written on row $y=0$ is undergoing the Collatz process with parity vector $p^\infty$. We know that this number is uniquely given by $z = Q^{-1}(p^\infty) \in \Z_2 \cap \Q$ \cite{10.2307/2322189,wirsching1998the}.
\end{proof}

\begin{example}\label{ex:cycle}
Figure~\ref{fig:cycle1}(c) shows an instance of Theorem~\ref{th:cycleconstruct} for parity vector $p = (1,1,0,0,0)$. On row $y=0$ (blue) we read the expansion $(00010110010)^\infty 00011 \in \Z_2$ (the outlined blue cells stop at then end of the first period of the expansion). As it is eventually periodic, it represents a rational number, but which one? In a general fashion, the following formula\footnote{The Python library \texttt{coreli} embeds functions to compute this formula, and to run the Collatz process on rational $2$-adic integers, see \url{https://github.com/tcosmo/coreli}.} holds to convert any eventually periodic $2$-adic integer to its associated fraction in $\Q$: let $w_0\in\{0,1\}^*$ be the initial segment and $w_1\in\{0,1\}^*$ be the infinitely repeated word then we have: $z = \ibin{w_0} + 2^{|w_0|}\ibin{w_1} \frac{1}{1-2^{|w_1|}} \in \Q$, see \cite{caruso:hal-01444183}. 
For the $3$-adic case: simply  replace $\{0,1\}^*$ by $\{0,1,2\}^*$, powers of $2$ by powers of $3$, and $\ibin{\cdot}$ by $\itr{\cdot}$. Here, the formula gives: $(00010110010)^\infty 00011 = \frac{5}{23} \in \Q$. This matches the fact that $\frac{5}{23} = Q^{-1}(p^\infty)$ indeed, the Collatz sequence of $\frac{5}{23}$ is: $(5/23,19/23,40/23,20/23,10/23,5/23,\ldots)$ whose parity vector is $p^\infty$ with $p = (1,1,0,0,0)$ (the parity of a rational $2$-adic is given by the parity of the numerator).
\end{example}

\subsubsection{The cyclic Collatz conjecture as a  \rcqca reachability problem}

Given Theorem~\ref{th:cycleconstruct}, the cyclic Collatz conjecture in $\N$ and $\Z$ can be reformulated as follows:

\begin{definition}[Parity vector cut]\label{def:pvc}
Let $p\in\Pa$. Consider $\mathcal{C} = \{ e_0, \dots , e_{|p|-1}\}$ the set of the positions of initial (half-defined) cells in $\tilde{c}_0[p] \in S^{\W/\!\sim_p}$ (Definition~\ref{def:initpt}). 
The parity vector cut associated to $v\in\Z^2$ is the set of positions  $\mathcal{C}_v = \{ e_0+v, \dots , e_{|p|-1}+v \} $.
\end{definition}

\begin{example}
Figure~\ref{fig:cycle1}(b) shows three parity vector cuts of $p=(1,1,0,0,0)\in\Pa$: all the cells underlined with the same color are in the same parity vector cut.
\end{example}

\begin{lemma}[Natural number and integer Collatz cycles]\label{lem:icc}
Let $p\in\Pa$ and $\tilde{c}_\infty[p] \in S^{\W/\!\sim_p}$ (Definition~\ref{def:initpt}).
\begin{enumerate}
\item $p$ is the support of a Collatz cycle with values in $\N$ if and only if there is $v\in\Z^2$ such that, in $\tilde{c}_\infty[p]$, all $|p|$ cells in the parity vector cut $\mathcal{C}_v$ are in state $(0,0)$ (Definition~\ref{def:pvc}).\label{pp:1}
\item $p$ is the support of a Collatz cycle with values in $\Z\setminus \N$ if and only if there is $v\in\Z^2$ such that, in $\tilde{c}_\infty[p]$, all $|p|$ cells in the parity vector cut $\mathcal{C}_v$ are in state $(1,1)$ (Definition~\ref{def:pvc}).\label{pp:2}
\end{enumerate}
Furthermore, the Collatz conjecture in $\N$ becomes: the only parity vectors for which Point~\ref{pp:1} holds are powers (in the sense of concatenation), and powers of rotations of $(0,1)\in\Pa$, corresponding to the cycle $(1,2,\dots)$.

The Collatz conjecture in $\Z\setminus \N$ becomes: the only parity vectors for which Point~\ref{pp:2} holds are powers and powers of rotations of $(1)$ and $(1,1,0) \text{ and } (1,1,1,1,0,1,1,1,0,0,0) \in\Pa$ which respectively correspond to the cycles $(-1,-1,\ldots)$, $(-5,-7,-10,-5,\ldots)$ and\\  $(-17, -25, -37, -55, -82, -41, -61, -91, -136, -68, -34, -17, \ldots)$.

\end{lemma}
\begin{proof}
Let $\mathcal{C}_v$ be the parity vector cut such that, for all $e\in\mathcal{C}_v$ we have $\tilde{c}(e)_\infty[p] = (0,0)$. By applying the equivalence relation $\sim_p$, we can translate the cut: there exists $x_0 < 0$ and $v' = (x_0,0)$ such that for all $e \in \mathcal{C}_v'$ we have $\tilde{c}(e)_\infty[p] = \tilde{c}(e+(v-v'))_\infty[p] = (0,0)$. Now, by application of the \rcqca local rule to cut $\mathcal{C}_v'$, we get that the cut $\mathcal{C}_v''$ with $v'' = (x_0-1,0)$ also contains only cells in state $(0,0)$. By induction, for all $x < x_0$, the cut $\mathcal{C}_{v_x}$ with $v_x = (x,0)$ contains only cells in state $(0,0)$. In particular, row $y=0$ starts with infinite segment of contiguous cells in state $(0,0)$. Hence, concatenating all its sum bits gives a $2$-adic expansion of the form $(0)^\infty w$ with $w\in\{0,1\}^*$, i.e. a natural number. By Theorem~\ref{th:cycleconstruct}, the Collatz cyclic number given by row $y=0$, $z=Q^{-1}(p^\infty)$ is in $\N$ and $p$ supports its cycle which gives Points~\ref{pp:1}. In the previous argument, replace state $(0,0)$ by $(1,1)$, prefix $(0)^\infty$ by $(1)^\infty$ and $\N$ by $\Z\setminus \N$ and we get the proof of Point~\ref{pp:2}.
\end{proof}
\begin{example}
Figure~\ref{fig:cycle2} shows the construction of all 4 non-zero integer Collatz cycles in the \rcqca, i.e. the cycles of $1,-1,-5$ and $-17$. For instance, in Figure~\ref{fig:cycle2}(c) the parity vector $p=(1,1,0)$, which supports the cycle $(-5,-7,-10,-5,\ldots)$ is inputted to the \rcqca. By Theorem~\ref{th:cycleconstruct}, on row $y=0$ we read the $2$-adic expansion of $-5$. Indeed, we read $z = (1)^\infty 011 \in \Z_2$ which is such that $z + (0)^\infty (101) = (0)^\infty$. Cells which are underlined in magenta are all in the first parity vector cut $\mathcal{C}_v$ (Definition~\ref{def:pvc}) with $v$ of the form $(x,0)$ with $x<0$ such that all cells in the cut are all in state $(1,1)$. By Lemma~\ref{lem:icc}, reaching such a parity vector cut is equivalent to saying that $p$ is the support of a stritcly negative cycle.
\end{example}

\begin{remark}
Note that, starting from $\tilde{c}(e)_0[p]$, cells in the parity vector cuts $\mathcal{C}_{v_x}$ with $v_x=(x,0)$ and $x<0$ could be constructed sequentially (by application of the \rcqca local rule) one cut after the other. First, this construction makes the fact that cyclic Collatz number are rational obvious: because each $\mathcal{C}_{v_x}$ is finite and the local rule deterministic, the states of cells on two different $\mathcal{C}_{v_x}$ will eventually be the same. This implies that each row of the cyclic world is eventually periodic in states and sum bits, hence, each row represents a rational $2$-adic: any Collatz cyclic number is rational. Second, in that sequential setting, the cyclic Collatz conjecture in $\N$ (resp. $\Z \setminus  \N$) becomes a reachability problem: from parity vector $p\in\Pa$, will a cut $\mathcal{C}_{v_x}$ with only states $(0,0)$ (resp. $(1,1)$) will ever be reached? 
\end{remark}

\begin{remark}
Known results \cite{wirsching1998the} on Collatz cycles imply that Point~\ref{pp:1} can hold only for $p$ such that $2^{|p|} \geq 3^{|p|_1}$ and Point~\ref{pp:2} can only hold for $p$ such that $2^{|p|} < 3^{|p|_1}$.
\end{remark}

\subsubsection{Infinite base conversion in the \rcqca}\label{sec:infbc}

Finally, it is an immediate consequence\footnote{In \cite{wirsching1998the}, it is shown that the denominator of a cyclic Collatz number is of the form $2^n - 3^m$ which can never be a multiple of $3$.} of \cite{wirsching1998the} that, for any $p\in\Pa$, we have that $Q^{-1}(p^\infty)$ is in $\Z_3$, i.e. $Q^{-1}(p^\infty)$ is both a $2$-adic and $3$-adic integer, i.e. it has both an infinite base 2 and an infinite base 3 integer representation. Hence, it is natural to interpret columns of $\tilde{c}_\infty[p]$ as elements of $\Z_3$ (i.e. infinite ternary expansion) in the exact same way as we interpreted columns in base $3$ in the finite case (Lemma~\ref{lem:cols}). We conjecture that Theorem~\ref{th:base_conversion} still holds:

\begin{conjecture}[Infinite base conversion]\label{cj:bc}
Let $p\in\Pa$ be a finite parity vector and let $\tilde{c}_\infty[p] \in S^{\W/\!\sim_p}$ be the limit configuration of the \rcqca when $p$ is given as input in the cyclic world $\W/\!\sim_p$ (Definition~\ref{def:initpt}). Let $z = Q^{-1}(p^\infty) \in \Z_2 \cap \Z_3 \cap \Q$ be the Collatz cyclic number associated to $p$ (Definition~\ref{def:ccn}).

Then, in $\tilde{c}_\infty[p]$, the $2$-adic expansion of $z$ is to be read on row $y=0$ (Theorem~\ref{th:cycleconstruct}) and the $3$-adic expansion of $z$ can be read by concatenating all base $3'$ symbols (and converting them, as usual, to ternary digits) on column $x=-|p|+1$.

Furthermore, for any anchor cell, i.e. cell at position $e$ such that both $\tilde{c}_\infty[p](e+\We)$ and $\tilde{c}_\infty[p](e+\No)$ are defined, the \emph{base conversion property} holds: the column strictly to the north of $e$ gives the $3$-adic expansion of $z\in\Z_2\cap \Z_3 \cap \Q$ of which 2-adic expansion is given strictly to the west of $e$.

\end{conjecture}

\begin{example}
Figure~\ref{fig:cycle1}(c) gives an example of Conjecture~$\ref{cj:bc}$. Indeed, the column in orange reads $(20021001011)^\infty 201 = \frac{5}{23}$ (see Example~\ref{ex:cycle} for the formula to compute fractions from eventually periodic $3$-adic expansions). That implies that the cell directly to the south of the orange column has the base conversion property: the row directly to its west (which is by equivalence in $\W/\!\sim_p$, the same as row $y=0$) gives the $2$-adic expansion of $\frac{5}{23}$ and the column directly to its north give the $3$-adic one (with base $3'$ symbols). In that sense, in the case of cyclic Collatz numbers, Theorem~\ref{th:base_conversion} is generalised to inifnite base 3/base 2 expansions.
\end{example}

\noindent {\bf Acknowledgement.} 
Sincere thanks to Olivier Rozier for fruitful interactions, and to anonymous reviewers for helpful comments. 
\vspace{-3ex}
\bibliographystyle{plain}
\bibliography{collatz2-main}
\else
\bibliographystyle{splncs04}
\bibliography{collatz2-main}
\fi

\if 0\mode
\appendix

\clearpage
\section{Additional figure}\label{app:gimmemore}
\begin{figure}[h!]

\begin{subfigure}[b]{0.475\textwidth}
\centering
\includegraphics[scale=1]{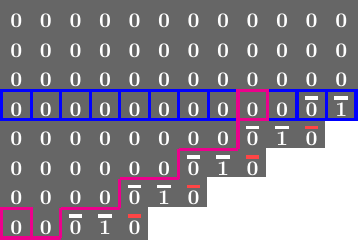} 
\caption{Cycle of $z=1$ as supported by $(1,0)^4 \in \Pa$.}
\end{subfigure}
\hfill
\begin{subfigure}[b]{0.475\textwidth}
\centering
\includegraphics[scale=1]{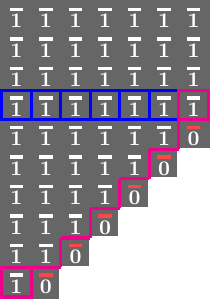} 
\caption{Cycle of $z=-1$ as supported by $(1)^6 \in \Pa$.}
\end{subfigure}
\vskip\baselineskip \vskip\baselineskip 
\centering
\begin{subfigure}[b]{0.475\textwidth}
\centering
\includegraphics[scale=1]{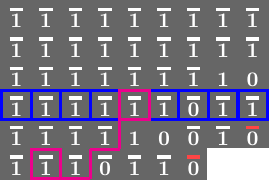} 
\caption{Cycle of $z=-5$ as supported by $(1,1,0) \in \Pa$.\vspace\baselineskip}
\end{subfigure}
\hfill
\begin{subfigure}[b]{0.45\textwidth}
\centering
\includegraphics[scale=0.7]{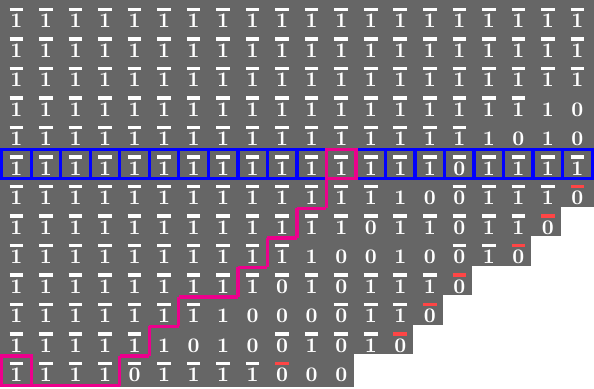} 
\caption{Cycle of $z=-17$ as supported by $(1,1,1,1,0,1,1,1,0,0,0) \in \Pa$.}
\end{subfigure}

\caption{Construction of the $4$ non-zero integer Collatz cycles in the \rcqca. For each $z\in\{1,-1,-5,-17\}$, a parity vector $p$, supporting the cycle of $z$ was inputted to the \rcqca in the cyclic world $\W/\!\sim_p$ and, in the limit configuration $\tilde{c}_\infty[p]$, the cycle is constructed (Definition~\ref{def:initpt} and Theorem~\ref{th:cycleconstruct}). Each panel of the figure highlights row $y=0$ (blue) which gives the $2$-adic extension of $z$ (Theorem~\ref{th:cycleconstruct}) as well as the first parity vector cut $\mathcal{C}_v$ (underlined in magenta, see Definition~\ref{def:pvc}) with $v$ of the form $v=(x_0,0)$ and $x_0 < 0$ such that all the cells in $\mathcal{C}_v$ are $(0,0)$ (cycle in $\N$) or $(1,1)$ (cycle in $\Z\setminus\N$), see Lemma~\ref{lem:icc}. Cells that are both outlined in magenta are at equivalent position in the cyclic world $\W/\!\sim_p$ (in particular they must contain have same state).
}\label{fig:cycle2}
\end{figure}

    \clearpage


\clearpage

\section{The $3x+1$ and dual $x/2$ finite state transducers}\label{app:fst}
In this appendix we formally define the two FSTs shown in Figure~\ref{fig:fst} and give some results about them.
\subsection{Definitions and duality}

Any affine transformation $ax+b$, with $a,b\in\N$, is computed, in any base, by a Finite State Transducer (FST), \cite{baseChanges}. Hence, it is not suprising that the $3x+1$ binary FST, Figure~\ref{fig:fst}(a), exists and, it was described in a similar way (but, merging states $\bar{0}$ and $1$) in \cite{sultanow2017}. However, what is intriguing is that, taking the dual of the $3x+1$ binary FST, gives rise to the FST, Figure~\ref{fig:fst}(b), which computes $x/2$ in base $3'$ (Lemma~\ref{lem:dual}) -- this duality phenomenon was unknown to us. Furthermore, the local rule of the \cqca is able to run both FSTs: horizontal applications of the rule simulate the $3x+1$ FST, Figure~\ref{fig:fst}(c), while vertical applications simulate the $x/2$ FST, Figure~\ref{fig:fst}(d). We recall that base $3'$ encodings are defined in Definition~\ref{def:btp}. 

Interestingly, base $3'$ acts like a bridge between base $2$ and base $3$, this is better explained by looking at the $3x+1$ and the $x/2$ FSTs (Figure~\ref{fig:fst}):

\begin{definition}[The $3x+1$ and $x/2$ FSTs]
\normalfont
\begin{enumerate}
\item The $3x+1$ FST has four states: $Q = \{0,\bar{0},1,\bar{1}\}$ and uses the two digit alphabet $\Sigma = \{0,1\}$.
Its transition function if given by the diagram in Figure~\ref{fig:fst}(a) and its initial state is $\bar 0$. It computes the map $3x+1$ in binary (Lemma~\ref{lem:fst3x}). 
\item The $x/2$ FST has two states: $Q' = \{0,1\}$ and uses the four symbol alphabet $\Sigma' = \{0,\bar{0},1,\bar{1}\}$. Its transition function if given by the diagram in Figure~\ref{fig:fst}(b) and its initial state is $0$. It computes the map $x/2$ in base $3'$ (Lemma~\ref{lem:fstx2}).
\end{enumerate}
\end{definition}

\begin{lemma}[Duality]\label{lem:dual}
\normalfont
The $3x+1$ FST and the $x/2$ FSTs are \emph{dual} to one another by which we mean that $Q = \Sigma'$, $Q'=\Sigma$ and $(q_1 \to q_2, s_1:s_2)$ is a transition in the $3x+1$ FST if and only if $(s_1 \to s_2, q_1:q_2)$ is a transition in the $x/2$ FST.\footnote{Note that this definition of duality does not specify any constraints on the choice of initial states.}
\end{lemma}
\begin{proof}
Immediate from Figure~\ref{fig:fst}.
\end{proof}

\begin{remark}
The local rule of the \cqca, which performs both the computation of the $3x+1$ FST (row by row) and of the $x/2$ FST (column by column), see Figure~\ref{fig:fst}(c) and (d), embeds the idea that, modulo duality, both FSTs (and their respective computation) are the same.
\end{remark}

\subsection{Correctness of computations}

Furthermore, these FSTs compute what their names suggest:

\lemT*
\begin{proof} 
We prove the following more general induction hypothesis on $n$, the size of the input word: $H(n):$ ``For $w\in\{0,1\}^n$, inputting $00w$ to the $3x+1$ FST, and starting at any initial state $q_0 \in \{0,\bar 0,1,\bar 1\}$, outputs the binary representation (with potential leading $0$s) of: $3\times \ibin{w} + \seum(q_0)$'', where we let $\seum(0) = 0, \seum(\bar 0) = \seum(1) = 1$ and $\seum(\bar 1) = 2$.

\textbf{Base case.} For $n=1$, following the diagram of Figure~\ref{fig:fst}(a) gives outputs $000, 001, 001, 010$ (representing $0,1,1,2$, and written right-to-left as usual) when inputting $000$ respectively from initial states $0,\bar 0, 1, \bar 1$. For input $001$, the following is output: $011, 100, 100, 101$  (representing $3,4,4,5$) when respectively starting from initial states $0,\bar 0, 1, \bar 1$. Hence, $H(1)$ holds. 

\textbf{Induction.} Let $n\in\N$ and let suppose that $H(n)$ holds. Let $w = w_{n}\dots w_0\in\{0,1\}^{n+1}$ and $x = \ibin{w}\in\N$. Let's input $00w$ to the $3x+1$ FST from initial state $q_0 = \bar 0$. Let's write $w'= w'_{n+2}\dots w'_{0}$ as the output of the FST. If $w_0 = 1$, then the second FST state is $q_1 = \bar{1}$ and $w'_0 = 0$ is output. Then, the word $w_{n}\dots w_1$ remains to be read, starting in state~$\bar{1}$. By the induction hypothesis, we get $\ibin{w'_{n+2}\dots w'_{1}} = 3\times \ibin{w_{n}\dots w_1} + 2 = 3 \frac{x-1}{2} + 2$. Hence, $\ibin{w'_{n+2}\dots w'_{0}} = 2\times\ibin{w'_{n+2}\dots w'_{1}} = 2\times(3 \frac{x-1}{2} + 2) = 3(x-1) + 4 = 3x + 1$. Which is what we wanted. All the other cases (different $q_0$ and $w_0$) are treated similarly. 
\end{proof}

\begin{remark}
\label{rk:z2}
Note that, a more intuitive way to describe the behvior of the $3x+1$ FST is the following: it sums each bit of the input to its right neighbor and the potential carry placed on it (encoded in the state of the FST). That corresponds to the binary addition $x+2x = 3x$. The initial state gives the border condition: $0$ for $3x + 0$, $\bar 0$ and $1$ for $3x+1$, and $\bar 1$ for $3x+2$. That remark leads to the fact that inputting infinite binary words to the $3x+1$ FST leads to computing the $3x+1$ map in $\Z_2$, the ring of $2$-adic integers. Indeed, in $\Z_2$, the rules of computation are the exact same as for finite binary strings \cite{caruso:hal-01444183} and the above interpretation of the $3x+1$ operation will generalise effortlessly. Hence, Definition~\ref{def:icw} can be generalised to handle infinite input words $w\in\{0,1\}^\N$ and their initial and limit \cqca configuration $c_0[w],c_\infty[w]\in S^\W$ and Lemma~\ref{lem:rows} can extend to showing that infinite rows of the \cqca simulate the Collatz process in $\Z_2$.

\end{remark}

\lemD*
\begin{proof}
We prove the following more general induction hypothesis on $n$, the size of base $3'$ representations: 
$H(n):$ ``Let $\gamma$ be the base $3'$ representation of $x\in\N$ potentially including leading $0$s, such that $|\gamma| = n$. Then, inputting $\gamma$ to the $x/2$ FST from initial state $q_0 = 0$ outputs $\gamma'$ such that: $x = 2\times \itp{\gamma'} + q_{|\gamma|}$ with $q_{|\gamma|}$, the state of the FST immediately after reading all the symbols of $\gamma$. Starting from initial state $q_0 = 1$ outputs $\gamma'$ such that: $x+3^n = 2\times \itp{\gamma'} + q_{|\gamma|}$.''. Note that contrary to the $3x+1$ FST, the $x/2$ FST reads inputs and produces outputs from the most significant symbol to the least significant.

\textbf{Base case.} Reading Figure~\ref{fig:fst}(b) gives: for $q_0=0$ and inputs $0,\bar 0, \bar 1$, the outputs are respectively $0,0,1$ and final states respectively $0,1,0$, as expected: the output is the quotient of division by $2$ and final state the remainder. For $q_0 = 1$ and inputs $0,\bar 0, \bar 1$, the outputs are respectively $1,\bar{1},\bar{1}$ and final states respectively $1,0,1$. Because $n=1$ and $3^n=3$, this respectively interprets as: $0 + 3 = 2\times 1 + 1$, $1 + 3 = 2\times2 + 0$ and $2 + 3 = 2\times 2 + 1$, as expected. Note that we don't have to check for base 3$'$ input $1$ as it is not a valid base $3'$ representation of an integer (Definition~\ref{def:btp}). Hence, $H(1)$ holds.

\textbf{Induction.} Let $n\in\N$ and let's suppose that $H(n)$ holds. Let $\gamma\in\{0,\bar 0, 1, \bar 1\}^*$ be the base $3'$ representation  of a natural number $x\in\N$ (Definition~\ref{def:btp}) such that $|\gamma| = n+1$. We write $\gamma = \gamma_{n} \dots \gamma_{0}$. Let's suppose $\gamma_{n} = \bar 0$ and let's input $\gamma$ to the $x/2$ FST starting on state $q_0 = 0$, we write $\gamma' = \gamma'_{n} \dots \gamma'_{0}$ as the output. After the first step, by reading Figure~\ref{fig:fst}(b), we get $\gamma'_{n} = 0$ and $q_1 = 1$, the next FST state. By the induction hypothesis, because $\gamma_{n-1} \dots \gamma_{0}$ is a valid base $3'$ representation of an integer (immediate from Definition~\ref{def:btp}), we get that $2\times \itp{\gamma'_{n-1} \dots \gamma'_{0}} + q_{n+1} = 3^n + \itp{\gamma_{n-1} \dots \gamma_{0}}$ with $q_{n+1}$ the sate after reading the $n$ symbols of $\gamma_{n-1} \dots \gamma_{0}$ (equivalenty, the $n+1$ symbols of $\gamma$). Hence, because $\gamma'_{n} = 0$, we get $2\times \itp{\gamma'} + q_{n+1} = 0*3^{n+1} + 3^n + \itp{\gamma_{n-1} \dots \gamma_{0}} = 3^n + \itp{\gamma_{n-1} \dots \gamma_{0}}$. But, since $\gamma_{n} = \bar 0$, we have $3^n + \itp{\gamma_{n-1} \dots \gamma_{0}} = \itp{\gamma_{n} \dots \gamma_{0}}$. Hence, we get $2\times \itp{\gamma'} + q_{n+1} = \itp{\gamma}$ which is what we want. All the other cases (different $q_0$ and $\gamma_{n}$) are treated similarly.

Hence, $H(n)$ holds for all $n$. In particular, that implies the result of this lemma: when we input the base $3'$ representation of $x\in\N$, the $x/2$ FST outputs $\gamma'$, from most significant to least significant symbol, such that: $\itp{\gamma'} = \lfloor \frac{x}{2} \rfloor$ and the final state of the $x/2$ FST gives the parity of~$x$.
\end{proof}

\begin{remark}
Note that the proof of Lemma~\ref{lem:fstx2} shows that, when $\gamma$, the base $3'$ representation of some $x\in\N$ (with potential leading $0$s) is inputted to the $x/2$ FST, starting on initial state $1$ instead of $0$, the following output is computed after $|\gamma|$ steps: $\lfloor \frac{3^{|\gamma|}+x}{2} \rfloor$. For $x$ odd, this operation corresponds to $\frac{3^{|\gamma|}+x}{2}$ which, if $x$ is not a multiple of $3$, is the expression of $2^{-1} * x$ in the multiplicative group $(\Z/3^k \Z)^*$. Said otherwise, the $x/2$ FST is still dividing by $2$ but, in modular arithmetic. 
\end{remark}

\begin{remark}\label{rk:inv}
In terms of reversibility, duality gives some additional structure to the FSTs. The $3x+1$ FST is reversible: inverting read/write instructions leads to a deterministic FST. 
The $x/2$ FST is not reversible but, the dual of the inverse of the $3x+1$ FST is the $x/2$ FST of which arrows are inverted. More generally, inverting read/write instructions of one FST corresponds to inverting arrows in the other.
\end{remark}

\begin{remark}
It would be interesting to explore whether, for any pair of such \emph{dual} FSTs, a 2D automaton, similar to the \cqca which runs both FSTs at the same time (one in each direction), can be constructed or not and if problems similar to the Collatz conjecture or Lagarias' periodicity conjecture \cite{10.2307/2322189} can be stated.
\enlargethispage{\baselineskip}
\enlargethispage{\baselineskip}
\end{remark}

\fi

\end{document}